\documentclass[12pt]{article}

\usepackage{amsmath}
\usepackage{amssymb}
\usepackage{amsthm}
\usepackage{amsfonts}
\usepackage{amscd}
\usepackage{graphicx}
\usepackage{geometry}
\usepackage{subfigure}
\usepackage{xypic}

\newtheorem{theorem}{Theorem} 
\newtheorem{definition}[theorem]{Definition}
\newtheorem{example}[theorem]{Example}
\newtheorem{lemma}[theorem]{Lemma}
\newtheorem{prop}[theorem]{Proposition}
\newtheorem{cor}[theorem]{Corollary}
\newtheorem{remark}[theorem]{Remark}

\newcommand{\reals}{\mathbb{R}}
\newcommand{\R}{\mathbb{R}}

\newcommand{\naturals}{\mathbb{N}}

\newcommand{\pos}{\operatorname{pos}}
\newcommand{\ra}{\rightarrow}
\newcommand{\intersect}{\cap}
\newcommand{\transpose}{\top} 
\newcommand{\im}{\operatorname{im}}
\newcommand{\interior}{\operatorname{int}}

\newcommand{\ind}[1]{\mathbb{I}\{#1\}} 
\newcommand{\pr}[1]{P\left(#1\right)} 
\newcommand{\tildepr}[1]{\tilde{P}\left(#1\right)} 
\newcommand{\ex}[1]{\mathbb{E}[#1]} 
\newcommand{\var}[1]{\operatorname{Var}(#1)}

\newcommand{\corr}[1]{\operatorname{Corr}(#1)}

\newcommand{\cbr}[1]{\{#1\}}

\newcommand{\eqd}{\stackrel{d}{=}} 

\newcommand{\tFS}{{e^{\eta_{FS}}}}
\newcommand{\tF}{{e^{\eta_F}}}
\newcommand{\tS}{{e^{\eta_S}}}

\title{Graphical Models for Correlated Defaults}

\author{I.Onur Filiz\footnote{Department of Industrial Engineering and Operations Research, University of California Berkeley, CA 94720-1777.
 Email: onurf@ieor.berkeley.edu}\ \ Xin
Guo\footnote{Department of Industrial Engineering and Operations Research, University of California Berkeley, CA 94720-1777.
 Email: xinguo@ieor.berkeley.edu}\ \ Jason
Morton\footnote{Department of Mathematics, Stanford University Palo Alto, CA 94305.
 \mbox{Email: jason@math.stanford.edu}}\ \ Bernd
Sturmfels\footnote{Department of Mathematics, University of California Berkeley, CA 94720-3840. Email: bernd@math.berkeley.edu}
}

\begin{document}
\setkeys{Gin}{width=8cm}
\maketitle

\begin{abstract}

A simple graphical model for correlated defaults is proposed, with
explicit formulas for the loss distribution.
Algebraic geometry techniques are employed to show that this
model is well posed for default dependence: it represents any
given marginal distribution for single firms and pairwise correlation matrix. These
techniques also provide a calibration algorithm based on
maximum likelihood estimation. Finally, the model is compared with
standard normal copula model in terms of tails of the loss distribution
and implied correlation smile.

\end{abstract}

\section{Introduction}

Credit risk concerns the valuation and hedging of defaultable
financial securities. (See e.g.~Bielecki and Rutkowski~\cite{BR02}, Duffie and Singleton~\cite{DS03},
Bielecki et~al.~\cite{BJR03},
Lando~\cite{Lando04}, and the references therein).
Since investors almost always engage in a range of different
instruments related to multiple firms, successful
modeling of the interaction of default risk for multiple firms is
crucial for both risk management and credit derivative pricing.
The significance of default correlation
is highlighted by the  current financial crisis.

There are a number of approaches for modeling correlated default.
\mbox{Collin-Dufresne} et~al.~\cite{Collin-DufresneGoldsteinEtAl:03},
Duffie and Singleton~\cite{DuffieSingleton:99a} and
Sch\"onbucher and Schubert~\cite{SchonbucherSchubert:01}
extend the reduced form
models by assuming correlated intensity processes. Intensity-based
models, however, tend to induce
unrealistic levels of correlation.
Hull et~al.~\cite{HullPredescuWhite:06}, Hull and White~\cite{HullWhite:01}, and Zhou~\cite{Zhou:01b} take the
structural form approach and use correlated asset processes,
extending the classical framework of Black and Cox~\cite{BlackCox:76}.
These models nevertheless imply spreads close to zero
for short maturities, similar to their single-firm counterparts.

Other approaches for default dependence include the so-called
``contagion models'', where default of
one firm affects the default process of the remaining firms. For
example, Davis and Lo~\cite{DavisLo:01} use binary random variables for the
default state of each firm, where these random variables are a
function of a common set of independent identically distributed
binary random variables. 
Jarrow and Yu~\cite{JarrowYu:01} extend the reduced form setup by assuming that the intensity for
the default process of each firm explicitly depends on the default
of other firms, thus one default causes jumps in intensities of
other firms' default processes. Giesecke and Weber~\cite{GieseckeWeber:06} place the
firms on the nodes of a multi-dimensional lattice, and model their interaction by
employing the {\em voter model} from the theory of interacting particle systems.
The top-down approach, on the other hand,
models the credit porfolio as a whole, focusing on the loss process rather than the processes of individual
firms. Some examples of this approach include
Giesecke and Goldberg~\cite{GieseckeGoldberg:05},
Errais et~al.~\cite{ErraisGieseckeEtAl:06},
Frey and Backhaus~\cite{FreyBackhaus}, Sch\"onbucher and Ehler~\cite{SchonbucherEhlers} and
Sidenius et~al.~\cite{SideniusEtAl}. These models present their strength in situations where the
modeling the individual firm process is not of primary importance:
modeling index reference portfolios or when the firms are very small
in comparison to the portfolio.

The binomial expansion method~\cite{CifuentesOConnor:96}, Credit
Suisse's CreditRisk$+$~\cite{CreditSuisse-CreditRiskPlus}, and
J.P.~Morgan's
CreditMetrics~\cite{GuptonFingerBhatia-CreditMetrics} are
well-known approaches in the finance industry. While BET represents the
loss distribution as a binomial random variable with the number of
trials in between the two extremes, CreditMetrics and
CreditRisk$+$ focus on individual defaults. Distribution of the
random variable representing the state of the firm is
parameterized by a set of factors which are shared among the
firms, but with varying weights. In contrast, copula models
(see Sch\"onbucher~\cite{Schonbucher:03} for a general survey,
and Li~\cite{Li:00}, Vasicek~\cite{Vasicek:87} for the normal
copula) separate the modeling of the interdependence of random
variables from the modeling of their marginal distributions.
Though popular due to its tractability,  normal copula suffers
from two well-recognized deficiencies: a) it fails to produce fat
tails observed in  the credit derivatives market for the
distribution of number of losses; b) the implied correlations in a
normal copula for the equity and senior tranches of a
Collateralized Debt Obligation(CDO) are higher than those for the
mezzanine tranches, a phenomenon known as  the ``correlation
smile''.

\paragraph{Our work}

In this  paper, a simple graphical model for
correlated defaults is proposed and analyzed (Section~\ref{sec:model_description}).  This model has an intuitive graphic
structure and the loss distribution for its special one-period version is simply a summation
of binomial random variables. This model is well posed
in capturing default dependence in the following sense:
it can  represent any
given marginal distribution for single firms and pairwise
correlation matrix. Techniques from
algebraic geometry are employed to prove this well-posedness
and to provide a calibration algorithm for the model.
Explicit formulas for the loss distribution and for CDO prices are
derived.
Finally, unlike the standard normal copula approach,
this model can produce fat
tails for loss distributions and  correct the correlation smile (Section~\ref{sec:analysis_comparison}).

In addition to the proposal and analysis of a simple model for
default correlation, one major contribution of this paper is the
introduction of a new algebraic technique to study  inequalities
implied in correlation structures. As correlation in any
multi-variate probability distribution naturally leads to certain
linear or non-linear inequalities, we are hopeful that this new
tool will provide a powerful alternative to existing approaches
such as copulas in the mathematical finance literature for
analyzing default correlation.

\section{The Graphical Model for Defaults}
\label{sec:model_description}

In this section a class of hierarchical models is formulated to
model default risk for multiple names. The most generic form
of the model is first presented, followed by a specialization
to homogeneous parameterizations for ease of
calibration and comparison with existing models.
To provide context, the simple terminology of
``firms'', ``sectors'' and ``default'' is adapted, although our
model is applicable in any generic context
with interaction between multiple entities.
In the finance context it includes any type of asset backed security (ABS) on multiple
names. For example,  one can represent a Collateralized Mortgage
Obligation (CMO) with our proposed graph structure, simply by  replacing
``sectors'', ``firms'' , and ``default'' with ``geographical
region'', ``mortgage holder'', and ``refinancing or default''
respectively.

\subsection{General Form} 
\label{ssec:general_form}

Take an undirected graph $G=(V,E)$ with $M$ nodes, and
denote the set of nodes by $V:=\cbr{1,\ldots, M}$. The edge set
$E$ is a subset of ${M \choose 2}$ possible pairwise connections
between any pairs of nodes, i.e. $E \subseteq
\cbr{(u,v): 1 \leq u < v \leq M }$. Each node of the
graph corresponds to a firm and has an associated binary random
variable $X_i, \, (i \in V)$ with $X_i = 1$ representing the
default of firm $i$ and $X_i = 0$ the survival. The joint
probability distribution of the random variable $X := (X_1,
\ldots, X_M)$ is given~by
\begin{eqnarray}
    p_w(\eta) :=\pr{X = w} &=& \frac{1}{Z} \cdot
     \exp\left( \sum_{i \in V}\eta_i w_i \,+ \sum_{(u,v) \in E}\eta_{uv} w_u w_v \right)
    \label{eq:Ising}
\end{eqnarray}
Here $w=(w_1,\ldots,w_M)$ runs over $ \cbr{0,1}^M$, the scalars
$ \eta_i \in \reals$ and $\eta_{uv} \in \reals$ are parameters,
and $Z$ is the
normalization constant known as the {\em partition function}:
$$  Z \quad = \,\,\sum_{w \in \cbr{0,1}^M} \exp\,\bigl(\, \sum_{i \in V}\eta_i w_i
\,+ \sum_{(u,v) \in E}\eta_{uv}w_u w_v \bigr),  $$

It is worth mentioning that Kitsukawa
et~al.~\cite{KitsukawaEtAl:06} and Molins and
Vives~\cite{MolinsVives:04}) have suggested using the long range
Ising model (LRIM) in the credit risk context. However, these
models are special cases of our formulation, and make use of physical concepts with no clear financial interpretation.  They restrict the structure to a single-period model with a completely connected graph and assume that all edge interactions are homogeneous.  We investigate heterogeneous connections and a sector model, analyze the multi-period setting, and provide pricing formulas.

\paragraph{Well-posedness of the model}

Probabilistic models of the form~(\ref{eq:Ising}) are also known
as {\em Markov random fields}, as
{\em Ising model} in physics,
or as {\em graphical models}
in computer science \cite{Jordan:99} and statistics
\cite{Lauritzen:96}.
In the finance context,
assessment of default correlation is usually assumed to identify
the following two sets of sufficient statistics:

\begin{itemize}
    \item The {\em marginal default probability}  $\pr{X_i=1}$
    is known for each firm $i \in V$.
    \item The {\em pairwise linear correlation}
        \begin{eqnarray}
            \label{eq:corrdefn}
            \rho_{uv}\,\,\,= \,\,\,\frac{\pr{X_u = X_v = 1} \,-\, \pr{X_u {=} 1}\cdot \pr{X_{v} {=} 1}}
            {\sqrt{\pr{X_u{=}1} \cdot \left( 1 -
            \pr{X_u {=} 1} \right)\cdot \pr{X_v {=} 1} \cdot \left( 1 - \pr{X_v {=}1} \right)}}
        \end{eqnarray}
        is assumed to be known for all pairs of firms
        $u$ and $v$ that share an edge in $E$.
\end{itemize}
Therefore, we shall demonstrate that this model is well posed for
modelling correlated default:  for every set of marginal default
probabilities and correlations, there exists a unique set of
parameters matching that information.

Clearly, data on the marginal default probabilities and the pairwise
linear correlations is equivalent to the following
set of $M+|E|$ sufficient statistics:
\begin{itemize}
    \item The single node marginals $\,\,P_i := \pr{X_i = 1} \quad $ for all $i \in V$.
    \item The double node marginals $\,\,P_{uv} := \pr{X_u = X_v = 1}\,\,$ for all $\,(u, v) \in E$.
\end{itemize}
Denoting this set of marginals by
 $\,P_\bullet \in {[0,1]}^{M+|E|}$, we shall show:
\begin{theorem}
    \label{prop:representation}
    Assume any given set of statistics $P_\bullet$ from {\em some} probability distribution on $M$ binary random
    variables.
    Then, there exists a unique set of parameters $\eta_i, \eta_{uv}$ such that the single
    and double node marginals implied by Equation~(\ref{eq:Ising}) match $P_i, P_{uv}$.
\end{theorem}

The proof of Theorem~\ref{prop:representation} relies on
techniques from algebraic geometry.
The key ingredients of the proof are illustrated through a simple example to gain
some insight. These ingredients are essential for model
implementation as well (see Section
\ref{sec:parameter_estimation}).

The first ingredient is an integer matrix $A_G$ associated with
the graph model.

\begin{example}
\label{ex:triangle} \rm
Let $G$ be the triangle with $\,V = \{1,2,3\}\,$and $\,E =
\bigl\{\{1,2\}, \{1,3\}, \{2,3\}\bigr\}$. The marginals $P_i$,
$P_{uv}$ are characterized by the following $16$ linear equalities:
\begin{eqnarray}
& P_i \geq P_{ij} \geq 0 \,\,\,\hbox{for all} \,\, i,j  , \nonumber \label{ineq1}\\
& P_1+P_2 \,\leq\, P_{12} + 1\,,\,\,
P_1+P_3 \,\leq\, P_{13} + 1\,,\,\,
P_2+P_3 \,\leq \,P_{23} + 1\,, \label{ineq2} \\
&  \, P_1 + P_{23} \,\leq \, P_{12} + P_{13}  \,,\,\,
 \, P_2 + P_{13} \,\leq \, P_{12} + P_{23}  \,,\,\,
 \, P_3 + P_{12} \,\leq \, P_{13} + P_{23}  \,,
\label{ineq3} \nonumber \\
& {\rm and} \quad
P_1+P_2+P_3 \,\leq\, P_{12}+P_{13}+P_{23} + 1.
\label{ineq4} \nonumber
\end{eqnarray}
These inequalities can be derived in the following way. First
consider the expansion of marginal default probabilities in terms
of elementary probabilities $p_{000}, p_{001}, \ldots, p_{111}$:
\begin{displaymath}
	P_i \,\,= \sum_{w \in \cbr{0,1}^3} a_{iw} p_w
\end{displaymath}
where $a_{iw} \in \cbr{0,1}$. Then construct a  $\cbr{0,1}$ valued
matrix $A_G$  using the $a_{iw}$ values, where each row
corresponds to a marginal probability, whereas the columns
correspond to the elementary probabilities. For the example, this
matrix becomes
\begin{equation} \label{trianglematrix}
\bordermatrix{
        & p_{000} & p_{001} & p_{010} & p_{011} & p_{100} & p_{101} & p_{110} & p_{111} \cr
P_1 &   0   &   0    &   0   &  0    &   1    &   1    &   1    &   1  \cr
P_2 &   0   &   0    &   1   &  1    &   0    &   0    &   1    &    1 \cr
P_3 &   0   &   1    &   0   &  1    &   0    &   1    &   0    &    1 \cr
P_{12} & 0 &  0    &   0   &  0    &   0    &   0    &   1    &    1  \cr
P_{13} & 0 &  0    &   0   &  0    &   0    &   1     &   0    &    1  \cr
P_{23} & 0 &  0    &   0   &  1    &   0    &   0    &   0    &    1
   }
\end{equation}
Then we have
\begin{displaymath}
    \cbr{P_\bullet: P_\bullet \textrm{ satisfies inequalities~(\ref{ineq2})}}   \,=\,
 \cbr{P_\bullet\,:\, P_\bullet = A_G \cdot p , \,
p \in \reals_+^8 \,\,{\rm and} \! \sum_{w \in \cbr{0,1}^M}p_w = 1}.
\end{displaymath}
In other words, the solution set of the $16$ linear inequalities in
(\ref{ineq2}) is the six-dimensional polytope
which can be obtained by
taking the convex hull of the columns of $A_G$.

\end{example}

\begin{remark}  \rm
If $G$ is the complete graph on $M=4$ nodes then the corresponding
$10$-dimensional polytope is
described by $56$ facet-defining inequalities. In general,
the number of facets of this polytope grows at least exponentially
in $M$. See Wainwright and Jordan~\cite{WainwrightJordan:06} for an approach which
carefully addresses these  issues of complexity.
\end{remark}

In general, our graphical model can be represented as a {\em toric
model} as in Geiger et~al.~\cite{GeigerMeek} or Pachter and
Sturmfels~\cite[\S 1.2]{PachterSturmfels:05} by defining the
appropriate integer matrix $A_G$. This matrix represents the
linear map which takes the 
vector of elementary probabilities to the vector of marginals.
To be precise the matrix $A_G$ has $2^M$ columns and $M+|E|+1$ rows and
its entries are in $\{0,1\}$. The columns of $A_G$ are indexed by
the elementary probabilities
$$ p_w \,= \,p_{w_1 w_2 \cdots w_M} \,\, = \,\,  P(X_1 = w_1, \cdots , X_M = w_M) ,
\,\,\quad w = (w_1,\ldots, w_M) \in \{0,1\}^M. $$ All rows but the
last are indexed by the marginals $P_i$ for $i \in V$ and 
the correlations $P_{uv}$ for $\{u,v\} \in E$. 
The entries in these rows are the coefficients in
the expansion of the marginals in terms of the $p_w$. The last row
of $A_G$ has all entries equal to one, and it corresponds to
computing the trivial marginal $\,\sum_{w \in \{0,1\}^M} p_w = 1$.

To be consistent with the algebraic literature, we replace the model
parameters by their exponentials, thus obtaining new parameters
that are assumed to be positive:
$$ \theta_i := {\rm exp}(\eta_i) \,\,\,\hbox{for} \,\,i \in V \quad
\hbox{and} \quad
\theta_{uv} := {\rm exp}(\eta_{uv})
\,\,\,\hbox{for}\,\, \{u,v\} \in E. $$
The model parameterization~(\ref{eq:Ising}) now translates into the monomial form of \cite[\S 1.2]{PachterSturmfels:05},
\begin{equation}
\label{eq:toric}
 p_w \,\,\,= \,\,\,\frac{1}{Z} \cdot \prod_{i\in V} \theta_i^{w_i} \cdot \prod_{(u,v) \in E}
 \theta_{uv}^{w_u w_v} ,
 \end{equation}
 where the elementary probabilities are the monomials corresponding to the columns of $A_G$.
 The last row of $A_G$ contributes the factor $\frac{1}{Z}$.
 In multi-dimensional form, the function mapping parameters to the elementary probabilities is then defined as:
 \begin{equation}
	 f : \reals^{M+|E|} \rightarrow \reals^{2^M},
	 \quad \theta \mapsto \frac{1}{\sum_{j=1}^{2^M} \theta^{a_j}}\left(\theta^{a_1}, \cdots, \theta^{a_{2^M}}  \right)
	 \label{eq:toric_mapping}
 \end{equation}
 where $\theta^{a_j}:= \prod_{i=1}^{M+|E|} \theta_i^{a_{ij}}$.
  The model is the subvariety of the $(2^M-1)$-dimensional probability
 simplex cut out by the binomial equations
 $$  \prod_w p_w^{C_w} \,-\,  \prod_w p_w^{D_w} \,\,\, = \,\,\, 0 ,$$
 where $C,D$ run over pairs of vectors in $ \naturals^{2^M}$ such that $\,A_G \cdot C = A_G \cdot D$.

With the new notation, one can represent Example~\ref{ex:triangle}
in the following way.
$A_G$ is the $7 {\times} 8$-matrix obtained by augmenting
(\ref{trianglematrix}) with a row of ones. The model
parameterization (\ref{eq:toric}) leads to
\begin{eqnarray*}
& \bigl(\,p_{000}\,,\,p_{001}\,,\,p_{010}\,,\,p_{011}\,,\,p_{100}\,,\,p_{101}\,,\,p_{110}\,,\,p_{111}\bigr)
\qquad \qquad \quad \, \\
= &
\bigl( \,1,\, \theta_3\,, \,\theta_2\,, \,\theta_2 \theta_3 \theta_{23}\,,\,
\theta_1\,, \,\theta_1 \theta_3 \theta_{13}\, , \, \theta_1 \theta_2 \theta_{12}\,,\,
\theta_1 \theta_2 \theta_3 \theta_{12} \theta_{13} \theta_{23} \,\bigr)
\end{eqnarray*}
The model is the hypersurface in the seven-dimensional probability
simplex given by
\begin{equation}
\label{quarticbino}
 p_{000} p_{011} p_{101} p_{110} \,\, = \,\, p_{001} p_{010} p_{100} p_{111}.
 \end{equation}

Next, with this new notation, we introduce the second ingredient
of the proof: Birch's theorem, which implies that this
six-dimensional toric hypersurface
(\ref{quarticbino}) is mapped bijectively onto the
six-dimensional polytope  (\ref{ineq2})
under the linear map $A_G$.

\begin{theorem}[Birch's theorem]
Every non-negative point on the toric
variety specified by an integer matrix $A$ is mapped bijectively
onto the corresponding polytope ${\rm conv}(A)$ under the linear
map $A$.
\end{theorem}
A proof of Birch's theorem can be found in
Appendix~\ref{app:birch}. For a more complete treatment, see
e.g.~\cite[Theorem 1.10]{PachterSturmfels:05}.
Now we are ready to prove Theorem  \ref{prop:representation}.

\begin{proof}[Proof of Theorem \ref{prop:representation}]
The linear map $A_G$ maps the $(2^M-1)$-dimensional probability
simplex onto the convex hull ${\rm conv}(A_G)$ of the column
vectors of $A_G$. The
mapping $A \rightarrow conv(A)$ is usually referred to as the {\em
marginal map of a log-linear model} in  statistics
(Christensen~\cite{Christensen:90}), or {\em moment map in toric
geometry} in mathematics (Fulton~\cite[\S 4]{Fulton1993}).
The convex polytope ${\rm conv}(A_G)$ therefore consists
of all vectors of marginals that arise from {\bf some} probability
distribution on $M$ binary random variables. Now applying {\em
Birch's theorem} to  the matrix $A_G$ yields the assertion of the
theorem.
\end{proof}

As a corollary, we conclude by the Main Theorem
for Polytopes~\cite[Theorem 1.1, page 29]{Ziegler:1995} that the possible marginals $P_\bullet$
arising from Equation~(\ref{eq:corrdefn})
are always characterized by a finite set of linear inequalities
as in (\ref{ineq2}).
We also note that the above techniques, especially Birch's theorem, are
instrumental for model calibration.

\paragraph{Calibration}
\label{sec:parameter_estimation}

There are several
algorithms for finding unique model parameters matching
any given set of marginal default probabilities and correlations
under the general formulation of Equation~(\ref{eq:Ising}).
The calibration problem is equivalent to maximum likelihood
estimation for toric models.
Indeed, suppose that one is given a
data vector $u \in \naturals^{2^M}$ whose coordinates specify how
many times each of the states in $\{0,1\}^M$ was observed. This
data gives rise to an empirical probability distribution with
empirical marginals $\,\frac{1}{\sum_{w=1}^{2^M} u_w} A_G \cdot u \, = \,
(P_\bullet)$, where $A_G$ is defined as in Section~\ref{ssec:general_form}.
The likelihood function of the data $u$ is the
following function of model parameters:
\begin{equation}
	\R^{M+|E|}\, \rightarrow \R_{>0} \,,\,\,\,
	\eta\, \,\mapsto \prod_{w \in \{0,1\}^M} \!\! p_w(\eta)^{u_w} \label{eq:likelihood_function}.
\end{equation}
Here $p_w(.)$ is defined as in Equation~(\ref{eq:Ising}).
Thus, a direct consequence of Theorem \ref{prop:representation} is
\begin{cor}
\label{cor:max_likelihood}
The likelihood function~(\ref{eq:likelihood_function}) has a unique maximizer $\hat{\eta}$.
This is the unique parameter vector whose probability distribution
implied by Equation~(\ref{eq:Ising}) has the empirical marginals $(P_\bullet)$.
\end{cor}
The key idea in the proof (given in
Appendix~\ref{app:max_likelihood}) implies that
computing the maximum likelihood parameters amounts to solving the following
optimization problem:
\begin{eqnarray}
	&&\max_{\eta}\,\, { {-\sum}_{w \in \{0,1\}^N} p_w \cdot {\rm log}(p_w)} \label{eq:objective} \\
	\textrm{s.t.}&& p_w = 
		\exp\left( \sum_{i \in V}\eta_i w_i \,+ \sum_{(u,v) \in E}\eta_{uv} w_u w_v \right)
		\quad \forall\, w \in \cbr{0,1}^M \label{eq:probability}\\
	&& A_G \cdot p = P_\bullet \label{eq:marginal}
\end{eqnarray}
Note that on the polytope of all probability distributions $p$
with constraint~(\ref{eq:marginal}),
the objective function~(\ref{eq:objective}) as a function of $p$ is strictly
concave with its maximizer $\hat p$ being the
distribution represented by $\hat \eta$. One can thus apply convex
optimization techniques to solve the
parameter estimation problem in our graphical model. In fact, this
optimization problem is also known as
{\em geometric programming}. See Boyd et~al.~\cite{Boyd:07} for an
introduction to this subject.

Parameter estimation in small toric models can be accomplished with
the {\em Iterative Proportional Fitting} of Darroch and Ratcliff \cite{DarrochRatcliff:72}; see Sturmfels~\cite[\S 8.4]{cbms} for an algebraic description and
a  {\tt maple} implementation. Such a straightforward
implementation of IPF requires
 iterative updates  of vectors with $2^M$ coordinates,
which is infeasible for larger values of $M$. To
remedy this challenge, one needs to turn to the large-scale
computational methods used in machine learning.  Popular methods aside from the convex optimization techniques mentioned above include those based on quasi-Newton methods such as LM-BFGS \cite{NocedalWright:99}, conjugate gradient ascent, log-determinant relaxation \cite{WainwrightJordan:06}, and local methods related to pseudolikelihood estimation (particularly in the sparse case).

\subsection{One Period Model}
\label{sec:sector}

For
both ease of exposition and numerical comparison with the existing
models in literature, we now investigate more specialized forms of the formulation.

First, we impose some structure on the graph. Take $M=N+S$, where nodes $1,
\dots, N$ represent individual firms and $N+1, \dots, N+S$ represent individual
industry sectors, so that the joint probability distribution for
$(X_1,\cdots,X_N)$ is defined as:
\begin{eqnarray}
    \label{eq:marginalization}
    \pr{X_1=x_1,\cdots,X_N=x_N}: = \sum_{s \in \cbr{0,1}^S} &&Q(X_1 = x_1, \cdots, X_N = x_N, \nonumber\\
    && X_{N+1} = s_1,\cdots,X_{N+S} = s_{N+S}).
\end{eqnarray}
Here the probability distribution $Q$ on the right hand side is
specified by Equation~(\ref{eq:Ising}).

Next, to capture the dependency among different industry sectors,
we specify the parameters $\eta_i$ and $\eta_{uv}$
as follows. We assume that each firm belongs to a particular
sector $j=1, 2, \cdots, S$ such that firm nodes $\cbr{1,\dots,N}$
are partitioned into $S$ subsets with $N_j$ elements, i.e.
$N=\sum_{j=1}^S N_j$. Moreover, a number of
homogeneity assumptions are imposed for simplicity:

\begin{itemize}
        \item Each firm node $i$ has a single edge, which connects to its
respective sector node.
    \item For any particular sector node $j$, all firm nodes that connect to it have the same node weight
        $\eta_{F_j}$ and same edge weight $\eta_{FS_j}$.
    \item Sector nodes are allowed to have different node weights
        $\eta_{S_j}$, and they can connect to each other with
        different edge weights $\eta_{N+u,N+v}$.
\end{itemize}
In short, the probability distribution for $(X_1,\cdots,X_N)$
in Equation~(\ref{eq:marginalization}) becomes
\begin{eqnarray}
    \label{eq:multiprob}
    \pr{X_1 = x_1, \dots, X_N = x_N} &=& \frac{1}{Z_S} \sum_{s \in \cbr{0,1}^S} \exp\left(
    \sum_{j=1}^S s_j \eta_{S_j} + s_j n_j \eta_{FS_j} + n_j \eta_{F_j} \right) \nonumber \\
    &&\exp\left(\sum_{(u,v): u,v \in \cbr{1,\cdots,S}} s_u s_v \eta_{N+u,N+v} \right)
\end{eqnarray}
where $n_j := \sum_{i:\eta_{i,N+j} \ne 0} x_i$ is the number of
defaulting firms in sector $j$, and $Z_S$ is the normalization constant.
Here
a sector random variable $X_{N+j}$ having value $0$ is interpreted as
that sector being financially healthy and $1$ as it being in
distress.

Note that if the graph $G$ breaks up into various connected
components, then the random variables associated with the nodes in
each component are independent of each other. This property allows
{\em conditional independence structure}s to be easily
incorporated into the model: when the state of all other other
firms is fixed, two firms not connected by an edge will default
independently of each other.  Also note that, by allowing
different parameters $\eta_{F_j}$ and $\eta_{{FS}_j}$ for each
sector, one can represent a diverse portfolio of firms,
with possibly negative pairwise default correlations.

Some simple calculation yields the loss distribution:
\begin{prop} Given the model specified by Equation~(\ref{eq:marginalization}),
\begin{eqnarray}
    \pr{\sum_{i=1}^N X_i=n} &=& \frac{1}{Z_S} \sum_{(n_1,\dots, n_S): \sum_j n_j = n} \sum_{s \in \cbr{0,1}^S}
    \exp\left(\sum_{(u,v): u,v \in \cbr{1,\cdots,S}} s_u s_v \eta_{N+u,N+v}
    \right) \nonumber \\
    &&\prod_{j} {N_j \choose n_j} \exp\left( s_j \eta_{S_j} + s_j n_j \eta_{FS_j} + n_j \eta_{F_j} \right),
    \label{eq:multiloss}
\end{eqnarray}
where
    \begin{eqnarray*}
    Z_S &=& \sum_{s \in \cbr{0,1}^S} \exp\left( \sum_{j=1}^S s_j \eta_{S_j} + \sum_{(u,v): u,v \in \cbr{1,\cdots,S}}
    s_u s_v \eta_{uv}\right) \prod_{j=1}^S \left( 1 + e^{\eta_{F_j}+ s_j \eta_{{FS}_{j}}}
    \right)^{N_j}.
    \label{eq:multiZ}
    \end{eqnarray*}
\end{prop}

\begin{figure}[h]
    \center
\[
   \begin{xy}<2cm,0cm>:
(2,0) *+![F]{\eta_{S_3}};
     p+(-.8,.4)  *++[o][F]{\eta_{F_3}} **@{.},
     p+(-.9,0)  *++[o][F]{\eta_{F_3}} **@{.},
     p+(-.8,-.4) *++[o][F]{\eta_{F_3}} **@{.},
     p+(-.4,-.8) *++[o][F]{\eta_{F_3}} **@{.},
p+(.5, .86) *+![F]{\eta_{S_1}} **@{-};
     p+(-.6,.6) *++[o][F]{\eta_{F_1}}  **@{.},
     p+(-.2,.8) *++[o][F]{\eta_{F_1}}  **@{.},
     p+(.2 ,.8) *++[o][F]{\eta_{F_1}}  **@{.},
     p+(.6,.6)  *++[o][F]{\eta_{F_1}}  **@{.},
p+(.5,-.86) *+![F]{\eta_{S_2}} **@{-};
     p+(.8,.4)  *++[o][F]{\eta_{F_2}}  **@{.},
     p+(.9,0)   *++[o][F]{\eta_{F_2}}  **@{.},
     p+(.8,-.4) *++[o][F]{\eta_{F_2}}  **@{.},
     p+(.4,-.8) *++[o][F]{\eta_{F_2}}  **@{.},
p+(-.86,0) **@{-};
(2.25,.43)  *+!R{\eta_{FS_{13}}} ;
(2.5,0)  *+!U{\eta_{FS_{23}}} ;
(2.95,.43)  *+!{\eta_{FS_{12}}} ;
(1.6,.2) *+{\eta_{FS_3}};
(2.8,1.16) *+{\eta_{FS_1}};
(3.2,-.4) *+{\eta_{FS_2}};
   \end{xy}
\]
    \caption{A $12$-firm graph with three sector nodes: Square nodes
    represent the sectors and circle nodes represent firms. }
    \label{fig:10firm3sector}
\end{figure}
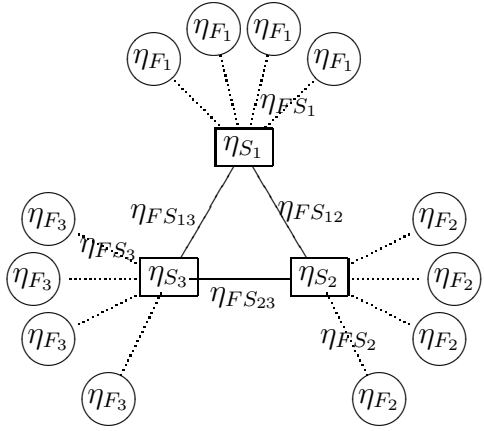

\paragraph{Connection to Binomial Distribution and Fat Tails}
Our model is related to binomial distribution with the simple
observation that the probability distribution in
Equation~(\ref{eq:multiloss}) can be decomposed into a summation of
$2 \cdot S$ independent binomial random variables. Indeed, note
that when $S=1$,
\begin{eqnarray}
    \pr{X_1=x_1,\cdots,X_N=x_N} &=& \frac{1}{Z_1} \left( e^{ \eta_F \sum_i x_i } + e^{ \eta_S +
    \left( \eta_{FS} + \eta_F \right) \sum_i x_i } \right), \label{eq:singleProb}
    \end{eqnarray}
    \begin{eqnarray}
    \pr{\sum_{i=1}^N X_i=n} &=& \frac{1}{Z_1} {N \choose n}\left[
    e^{n \eta_F} + e^{\eta_S + n \eta_F + n\eta_{FS}} \right] \label{eq:singleLoss} \\
    Z_1 &=& \left( 1 + e^{\eta_F} \right)^N + e^{\eta_S} \left( 1 + e^{\eta_F + \eta_{FS}} \right)^N, \label{eq:singleZ}
\end{eqnarray}
where $\eta_{FS_1}, \eta_{S_1}, \eta_{F_1}$ are replaced by
$\eta_{FS}, \eta_{S}, \eta_F$ respectively for notational
simplicity.
\begin{prop} Given the model (\ref{eq:marginalization}), when $S=1$,
\begin{equation}
    \sum_{i=1}^N X_i\eqd Y B_1 + (1-Y) B_2,
\end{equation}
where $Y,B_1,B_2$ are independent and distributed as
\begin{eqnarray*}
    Y &\sim&\textrm{Bernoulli}\left( \frac{\tS \left( 1 + \tF \tFS \right)^N}{Z_1} \right), \\
    B_1 &\sim& \textrm{Binomial}\left(\frac{\tF \tFS}{1+ \tF \tFS},N\right), \\
    B_2 &\sim& \textrm{Binomial}\left(\frac{\tF}{1+\tF},N\right).
\end{eqnarray*}
\label{lem:binomialdec}
\end{prop}
\begin{cor} Under the assumptions of Proposition~\ref{lem:binomialdec},
\begin{eqnarray}
    X_i &\eqd& Y R_i + ( 1 - Y ) U_i, \quad     {\rm with}
    \end{eqnarray}
    \begin{eqnarray*}
    R_i &\sim& \textrm{Bernoulli}\left( \frac{\tF \tFS}{1 + \tF \tFS} \right), \\
    U_i &\sim& \textrm{Bernoulli}\left( \frac{\tF }{1 + \tF }
    \right),
\end{eqnarray*}
and $R_1, \dots, R_N, U_1, \dots, U_N, Y$ are mutually
independent. Moreover,
\begin{eqnarray}
    \corr{X_1,X_2} &=& \corr{Y V_1 + U_1, Y V_2 + U_2} = \corr{Y V_1, Y V_2} \nonumber \\
    &=& \frac{\ex{Y^2}\ex{V_1}\ex{V_2} - \ex{Y}^2 \ex{V_1}\ex{V_2}}{\sqrt{\var{Y V_1}}\sqrt{\var{Y V_2}}}
    = \frac{\var{Y} \ex{V}^2}{\var{YV}} \nonumber\\
    &=& \frac{\var{Y\ex{V}}}{\var{YV}}
\end{eqnarray}
where $V \eqd V_i := R_i - U_i$ and independent of all other random
variables. \label{lem:correlation}
\end{cor}

One implication of Proposition~\ref{lem:binomialdec} is that
one can have control over the tails of the loss
distribution. Of the two binomial random variables, varying the parameters
affecting the center of the higher mean random variable to increase(decrease)
its mean results in thicker(thinner) tails.
Moreover, the loss distribution
may be bimodal. Bimodality can be explained by a `contagion' effect
among firms. Having a high number of defaults may make it more
likely for ``neighboring'' firms to default.
This phenomenon enables our model to
correct the so-called ``correlation smile'' (e.g.
Amato and Gyntelberg~\cite{AmatoGyntelberg:05}, Hager and
Sch\"obel~\cite{HagerSchoebel:06}) in pricing
CDOs, since a low probability for mezzanine level defaults
naturally lead to lower spreads for the respective tranche.
These will be illustrated in detail in Section~\ref{sec:copula}.

\subsection{Multi-period Model}
\label{ssec:multiperiodmodel}

In this section, we shall extend the one-period model to a
multi-period one. This extension is essential for pricing
defaultable derivatives and for comparison with standard copula
models.

The construction is as follows:
\begin{itemize}
    \item Start with a single-sector graph with $N$ firms. At each payment period $t_k$, the graph evolves by the
        defaulting of some nodes. Furthermore, some of the previously defaulted nodes are 
	removed. Economically, removal of nodes represents that
        these firms are no longer influencing or providing useful information about the default process of other firms.  Therefore, the number
        of firms remaining in the system is dynamic, and is denoted by $N_{t_k}$. Denote the number of firms
        that have defaulted up to $t_k$ by $D_{t_k}$. Then, $D_0=0$ and $N_0=N$.
    \item Each defaulted node ``stays'' in the system for a geometrically distributed number of time steps (with
        ``success'' or ``removal'' probability $p_R$), independent of everything else. This is equivalent to removing each
        defaulted node from the system with probability $p_R$, independent of everything else, at the beginning
        of $t_k$. Thus, the number of nodes that are currently in default 
        and still in the system at time $t_k$, $I_{t_k}$, is given by:
        \begin{displaymath}
            {I}_{t_k} = D_{t_k} + N_{t_k} - N
        \end{displaymath}
    \item Number of additional defaults $D_{t_{k+1}} - D_{t_{k}}$ during the period $(t_{k},t_{k+1})$ is based on a
        conditioning of the probability distribution specified by Equation~(\ref{eq:singleLoss}).  More specifically:
        \begin{eqnarray}
            &&\pr{D_{t_{k+1}} - D_{t_k} = n \mid {I}_{t_{k}} = m; N_{t_k}} \nonumber\\
            &:=& \sum_{i_1,\cdots,i_{n}: i_j \in \cbr{m+1,\cdots,N_{t_k}}}
            P\left(X_{i_1}=\cdots=X_{i_{n}}=1, X_{i_{n+1}} = \cdots = X_{i_{N_{t_k}}} = 0 \right. \nonumber\\
            && \left. \mid X_{1} = \cdots = X_{m} = 1 \right), \quad n + m \le N_{t_k}
            \label{eq:lossIncrementDistribution}
        \end{eqnarray}

\end{itemize}

This construction, together with some simple calculation, leads to
\begin{prop}
    \begin{displaymath}
        \pr{D_{t_{k+1}} - D_{t_k} = n \mid {I}_{t_{k}} = m; N_{t_k}} =
	\tildepr{n; N_{t_k} - m, \eta_S + m \eta_{FS}, \eta_{FS}, \eta_{F}}
    \end{displaymath}
    where $$\tildepr{n;N,\eta_S,\eta_{FS},\eta_F}: =\pr{\sum_{i=1}^N X_i=n}$$ as defined by
Equation (\ref{eq:singleLoss}) and
$$Z_{N,\eta_S,\eta_{FS},\eta_F}:=Z_1$$ as defined by
Equation (\ref{eq:singleZ}).
    \label{lem:lossIncrementExplicit}
\end{prop}

\subsubsection{Simulation and CDO pricing}

Based on the above proposition, it is easy to see that
this multi-period model can be simulated as follows. At
time $0$, the graph has $N$ non-defaulted firms, At time $t_1$,
no removal of nodes occurs since none of the firms were in default
at time $0$. The number of firms that default during period
$(0,t_1)$, $D_1 - D_0 = D_1$ is determined by sampling from
Equation~(\ref{eq:singleLoss}). At time $t_2$, each of the
$D_1$ nodes is removed from the graph with probability $p_R$. The
additional number of defaults $D_2 - D_1$ is determined by
Equation~(\ref{eq:lossIncrementDistribution}), where $N_2$ is the
total number of firms remaining after the removal, and
${I}_{t_2}$ is the number of firms among $D_1$ that have not been
removed from the graph. Continue in this fashion until $K$ periods
are covered.

Moreover, the homogeneity assumptions for the one-period
single-sector model imply that our model can be
perceived as a two-state discrete time Markov chain, as
$(D_{t_{k+1}},N_{t_{k+1}})$ only depends on $(D_{t_k}, N_{t_k})$.
Indeed, the transition matrix $P$ of the Markov chain is given by:
\begin{eqnarray}
    P_{(D_{t_k},N_{t_k})\rightarrow(D_{t_{k+1}},N_{t_{k+1}})} &=& { I_{t_k} \choose N_{t+k} - N_{t_{k+1}}} p_R^{N_{t_k} -
    N_{t_{k+1}}} {(1-p_R)}^{D_{t_k}-N+N_{t_{k+1}}} \nonumber \\
    && \tilde{P}(D_{t_{k+1}}-D_{t_k};N_{t_k}- {I}_{t_k},\eta_S + {I}_{t_k} \eta_{FS}, \eta_{FS}, \eta_F) \label{eq:transitionMatrix} \\
    && N \ge D_{t_{k+1}} \ge D_{t_k} \ge 0 \textrm{ and } N \ge N_{t_k} \ge N_{t_{k+1}} \ge 0 \nonumber
\end{eqnarray}
with $D_0=0, N_0=N$.
This Markov chain formulation is useful for analytical
calculation of loss distribution and CDO prices.

For purposes of model comparison in the later sections, we briefly
discuss pricing of CDOs in our model.

\paragraph{Collateralized Debt Obligation(CDO) Pricing}

A Collateralized Debt Obligation (CDO) is a portfolio of
defaultable instruments (loans, credits, bonds or default swaps),
whose credit risk is sold to investors who agree to bear the
losses in the portfolio, in return for a periodic payment. A CDO is
usually sold in {\em tranches}, which are specified by their {\em
attachment point}s $K_L$ and {\em detachment point}s $K_U$ as a
percentage of total notional of the portfolio. The
holder of a tranche is responsible for covering all losses in
excess of $K_L$ percent of the notional, up to $K_U$ percent. In
return, the premiums he receives are adjusted according to the
remaining notional he is responsible for. In the case of
popularly traded tranches on the North American Investment Grade
Credit Default Swap Index (CDX.NA.IG), the tranches are named {\em
equity, mezzanine, senior, senior, super-senior} with attachment
and detachment points of $0-3, 3-7, 7-10, 10-15, 15-30$
respectively.

Given an
underlying portfolio, and fixed attachment and detachment points
for all tranches, the pricing problem is the determination of
periodic payment percentages (usually called {\em spread}s)
$s_l$ for all tranches,
assuming the market is complete and
default-free interest rate is independent of the credit risk of
securities in the portfolio.

If we denote the total notional of the portfolio by $M$, the periodic payment dates by $t_1, \dots, t_K$, the date of
inception of the contract by $t_0:=0$, payment period $t_{k+1}-t_k$ by $\gamma$,
the total percentage of loss in the portfolio by time $t_k$ by $C_{t_k}$, the attachment (detachment) point for tranche $l$ by
$K_{L_l}$ ($K_{U_l}$), and the discount factor from $t_0$ to $t_k$ by $\beta(t_0, t_k)$, then it is clear that
specifying the distribution for $C_{t_k}$ for $k=1,\dots,K$ is sufficient
for pricing purposes.

To see this, note the percentage
of loss $C_{l,t}$ suffered by the holders of tranche $l$ up to time $t$ is given by:
\begin{eqnarray}
    C_{l,t} &:=& \min\cbr{C_{t}, K_{U_l}} - \min\cbr{C_{t}, K_{L_l}}.
    \label{eq:trancheLoss}
\end{eqnarray}
Consequently, the value at time $t_0$ of payments received by the holder of tranche $l$ is
\begin{eqnarray}
    &&\sum_{k=1}^K \beta(t_0,t_k) s_l \gamma M \ex{K_{U_l} - K_{L_l} - C_{l,t_k}}.
    \label{eq:trancheFixed}
\end{eqnarray}
Similarly, the value at time $t_0$ of payments made by the holder of tranche $l$ is given by
\begin{eqnarray}
    &&\sum_{k=1}^K \beta(t_0,t_k) M \ex{C_{l,t_k}-C_{l,t_{k-1}}}.
    \label{eq:trancheFloating}
\end{eqnarray}
In order to prevent arbitrage, the premium $s_l$ needs to be chosen such that the value of payments received is equal to
the value of payments made. Therefore,
\begin{eqnarray}
    s_l &=& \frac{\sum_{k=1}^K \beta(t_0,t_k) \left( \ex{C_{l,t_k}} - \ex{C_{l,t_{k-1}}}\right)}{
    \sum_{k=1}^K \beta(t_0,t_k) \gamma \left(K_{U_l} - K_{L_l} - \ex{C_{l,t_k}}\right)}.
    \label{eq:trancheSpread}
\end{eqnarray}

Now our focus is to calculate the distribution for $C_{t_k}$ in our multi-period model.
Denoting the $k$-step transition matrix (in
Equation~(\ref{eq:transitionMatrix})) for the Markov chain with $P^k$ and the
number of losses at the $k$-th step by $L^k$, then
\begin{equation}
\pr{C_{t_k} = \frac{m}{N}} = \pr{L^k = m} = \sum_{n=0}^N
P^k_{(0,N)\rightarrow(m,n)}, \label{eq:multiPeriodLossPercentage}
\end{equation}
and the spreads are given by
\begin{prop}
Given the yield curve $\beta$, attachment ($K_{L_l}$) and detachment  ($K_{U_l}$) points, and the implied Markov transition matrix $P$, the spread of tranche $l$ is
\begin{eqnarray*}
    s_l &=& \frac{\sum_{k=1}^K \beta(t_0,t_k) \left(
    \sum_{m = 0}^N \left(
    \min\cbr{\frac{m}{N}, K_{U_l}} - \min\cbr{\frac{m}{N}, K_{L_l}} \right)
    \sum_{n=0}^N P^k_{(0,N)\rightarrow(m,n)} - P^{k-1}_{(0,N)\rightarrow(m,n)}
    \right)}{
    \sum_{k=1}^K \beta(t_0,t_k) \gamma \left(K_{U_l} - K_{L_l} - \sum_{m = 0}^N \left(
    \min\cbr{\frac{m}{N}, K_{U_l}} -  \min\cbr{\frac{m}{N}, K_{L_l}} \right)  \sum_{n=0}^N P^k_{(0,N)\rightarrow(m,n)} \right)}.
\end{eqnarray*}
\end{prop}

\section{Sensitivity Analysis and Comparison with One-Factor Normal Copula}
\label{sec:analysis_comparison}

In this section, some numerical results on the proposed
model are reported, and they are compared with both the static(one-period)
and dynamic(multi-period) one-factor normal
copula model. Throughout the section, $S=1$ is assumed
for simplicity.

\subsection{Static Characteristics}

\paragraph{Correlation for single-period model}

First, we analyze the effects of the parameterization triplets
$(\eta_S,\eta_F,\eta_{FS})$ on the correlation between two firms
$corr(X_i,X_j)$. Note the following statement
\begin{prop}
    \label{thm:singleFindingMarginal}
    Given fixed $\eta_S, \eta_{FS}$ and $q \in (0,1)$, the $\eta_F^*$ value that gives $\pr{X_1=1}=q$ in
    Equation~(\ref{eq:singleProb}) is given by
    \begin{displaymath}
        g(e^{\eta_F^*}) + e^{\eta_S} g\left( e^{\eta_{FS}} e^{\eta_F^*} \right) =
        0,
    \end{displaymath}
    where
    \begin{displaymath}
        g(y) = \left( 1 - \frac{1-q}{q}y \right)\left( 1 + y
        \right)^{N-1}.
    \end{displaymath}
\end{prop}

Take $q=0.01,0.05$ and $N=125$.
Figures~\ref{fig:correlation_fixed_marginal_np}--\ref{fig:correlation_fixed_marginal_pn} show the
correlation values for the four quadrants on $\eta_S$ and $\eta_{FS}$. For each point
Proposition~\ref{thm:singleFindingMarginal} is utilized for calculating the $\eta_F^*$ value
that achieves the desired $q$ value.
Note that $\eta_{FS}$ is the dominant parameter when
$\eta_{FS}$ values are close to $0$. As $\eta_{FS}$ moves away
from $0$, $\eta_S$'s effect increases. Also note that it is
possible to obtain high degrees of linear correlation levels even for $q=0.01$.
This extends the abilities of multi-firm extensions of intensity-based models in literature
(see \cite[\S 10.5]{Schonbucher:03}  for a discussion).

\setkeys{Gin}{width=0.4\textwidth}

 \setkeys{Gin}{width=0.4\textwidth}
\begin{figure}[!h]
    \centering
    \subfigure[q=$0.01$] {
    \label{fig:correlation_1_np}
\includegraphics{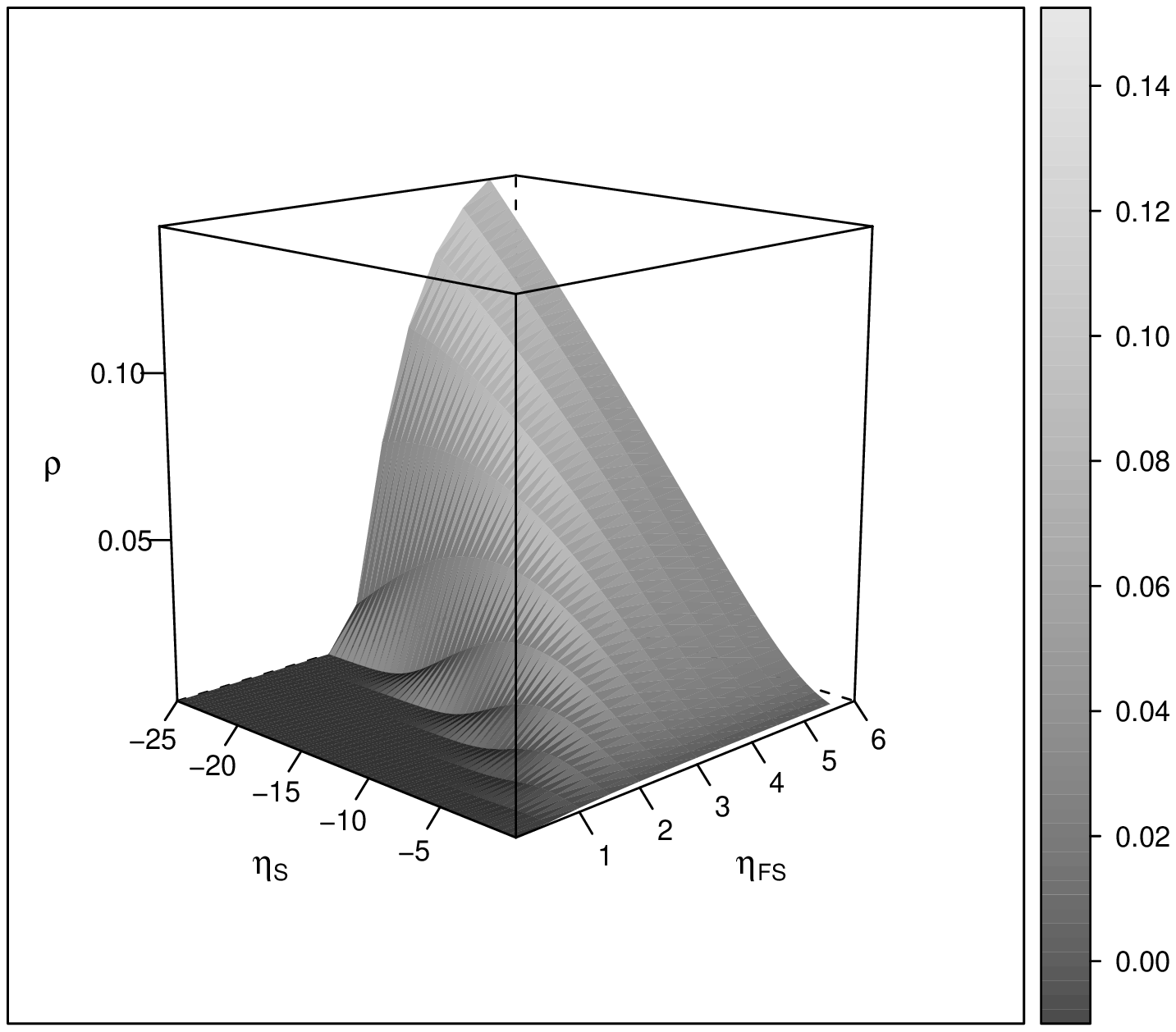}
    }
    \hspace{0.5in}
    \subfigure[q=$0.05$] {
        \label{fig:correlation_5_np}
\includegraphics{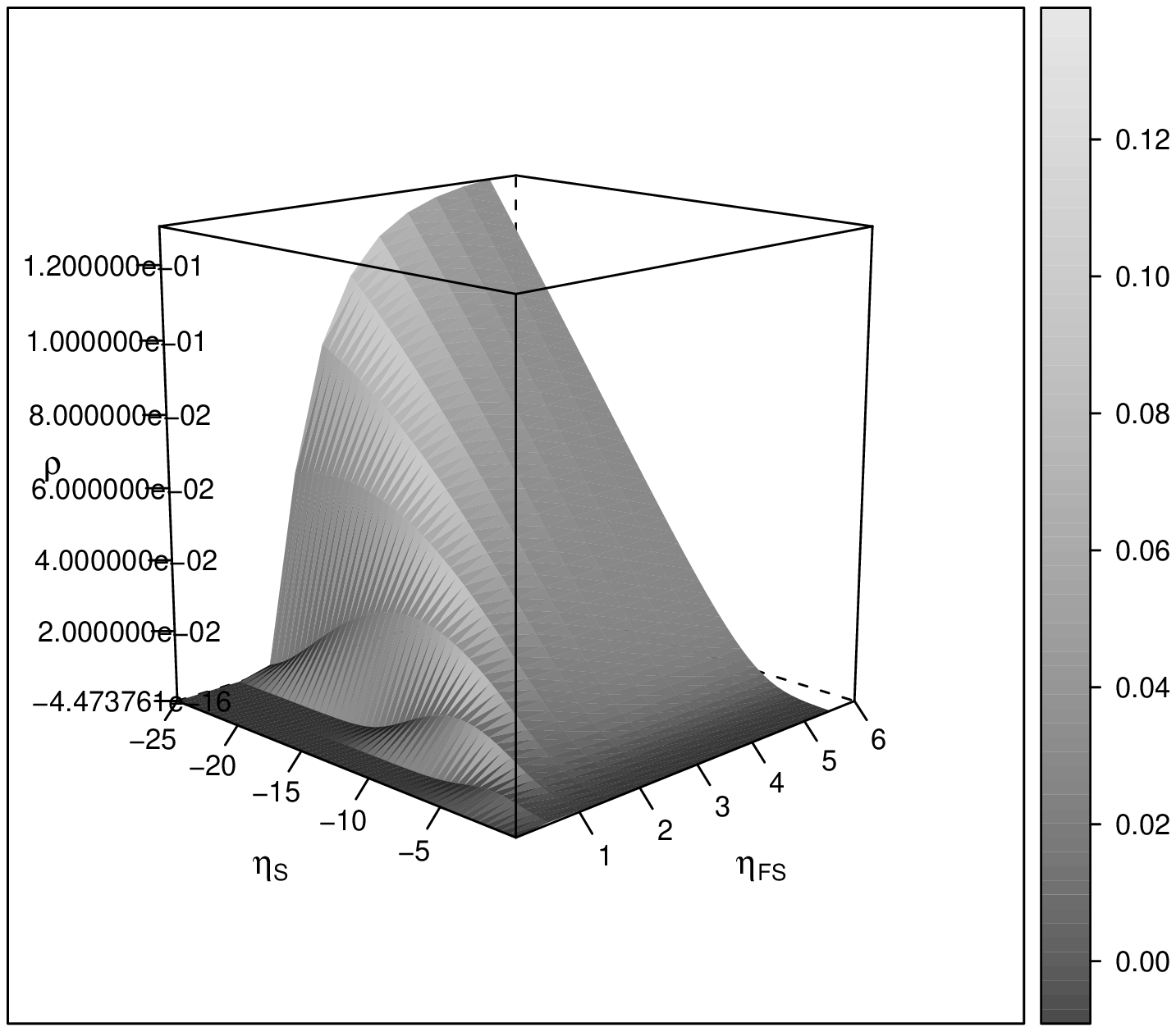}
    }
    \caption{Variation of $\rho$ for fixed marginal default
    probabilities, $\eta_S < 0, \eta_{FS} > 0$}
    \label{fig:correlation_fixed_marginal_np}
\end{figure}
\setkeys{Gin}{width=0.8\textwidth}

 \setkeys{Gin}{width=0.4\textwidth}
\begin{figure}[!h]
    \centering
    \subfigure[q=$0.01$] {
    \label{fig:correlation_1_nn}
\includegraphics{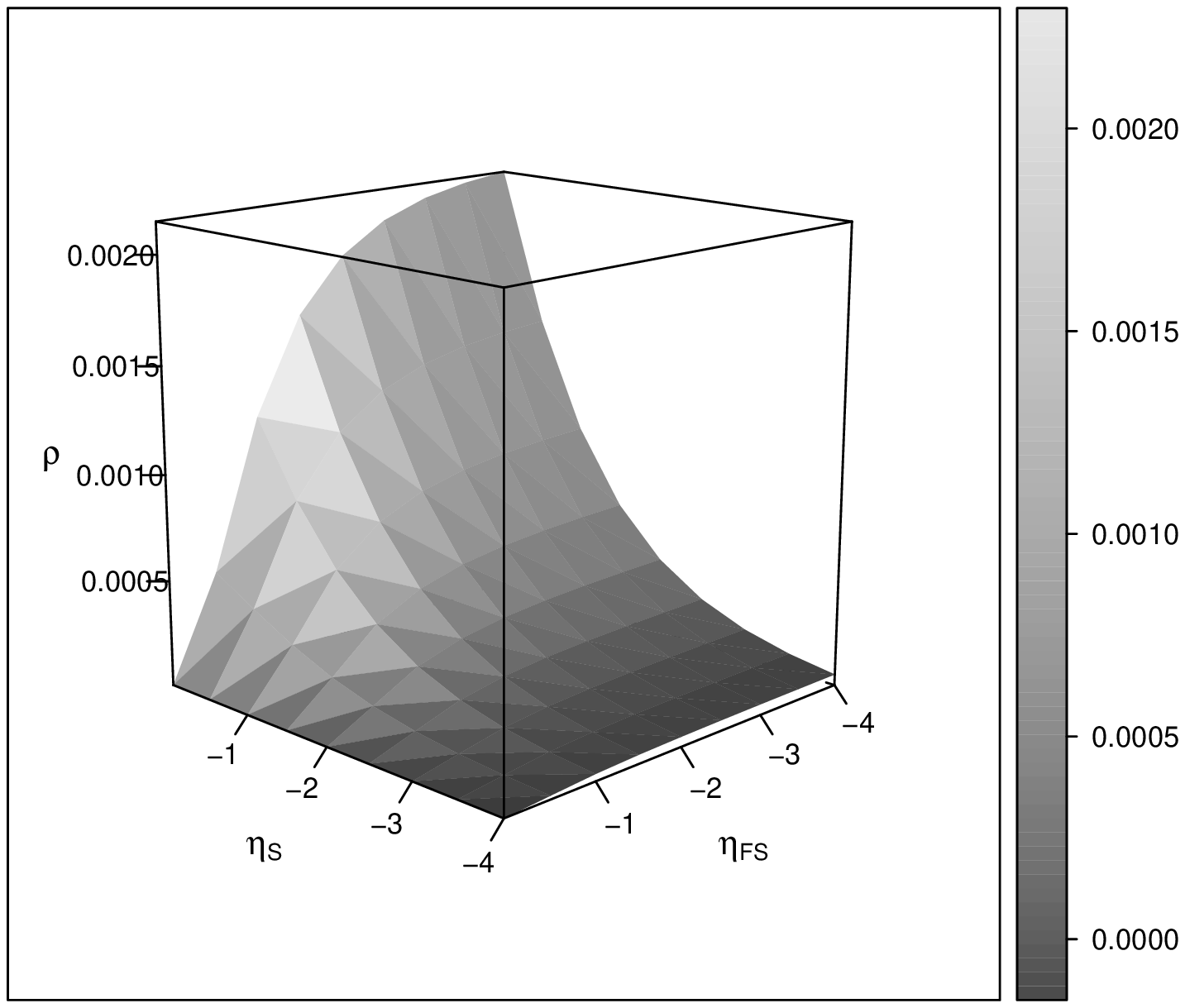}
    }
    \hspace{0.5in}
    \subfigure[q=$0.05$] {
        \label{fig:correlation_5_nn}
\includegraphics{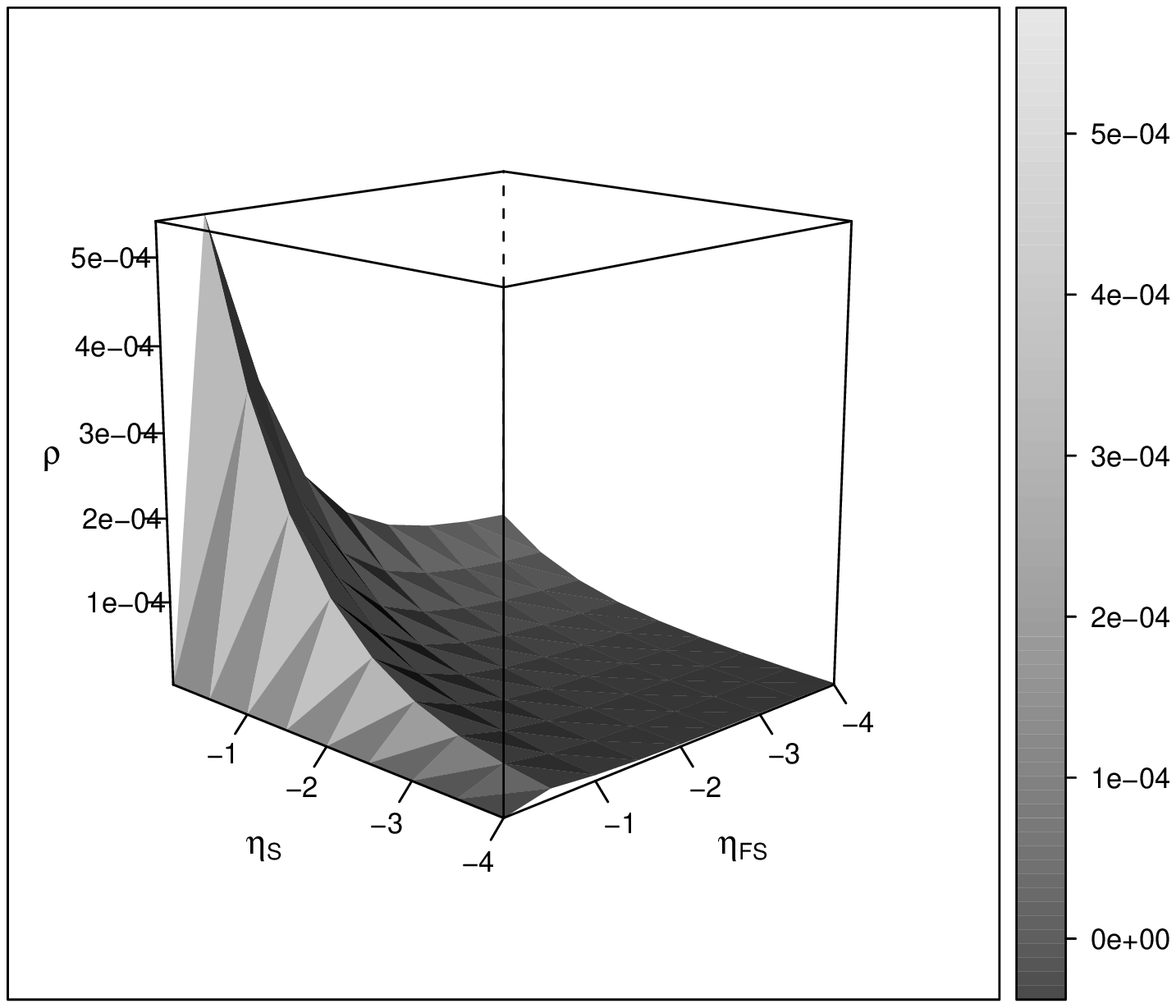}
    }
    \caption{Variation of $\rho$ for fixed marginal default
    probabilities, $\eta_S < 0, \eta_{FS} < 0$}
    \label{fig:correlation_fixed_marginal_nn}
\end{figure}
\setkeys{Gin}{width=0.8\textwidth}

 \setkeys{Gin}{width=0.4\textwidth}
\begin{figure}[!h]
    \centering
    \subfigure[q=$0.01$] {
    \label{fig:correlation_1_pp}
\includegraphics{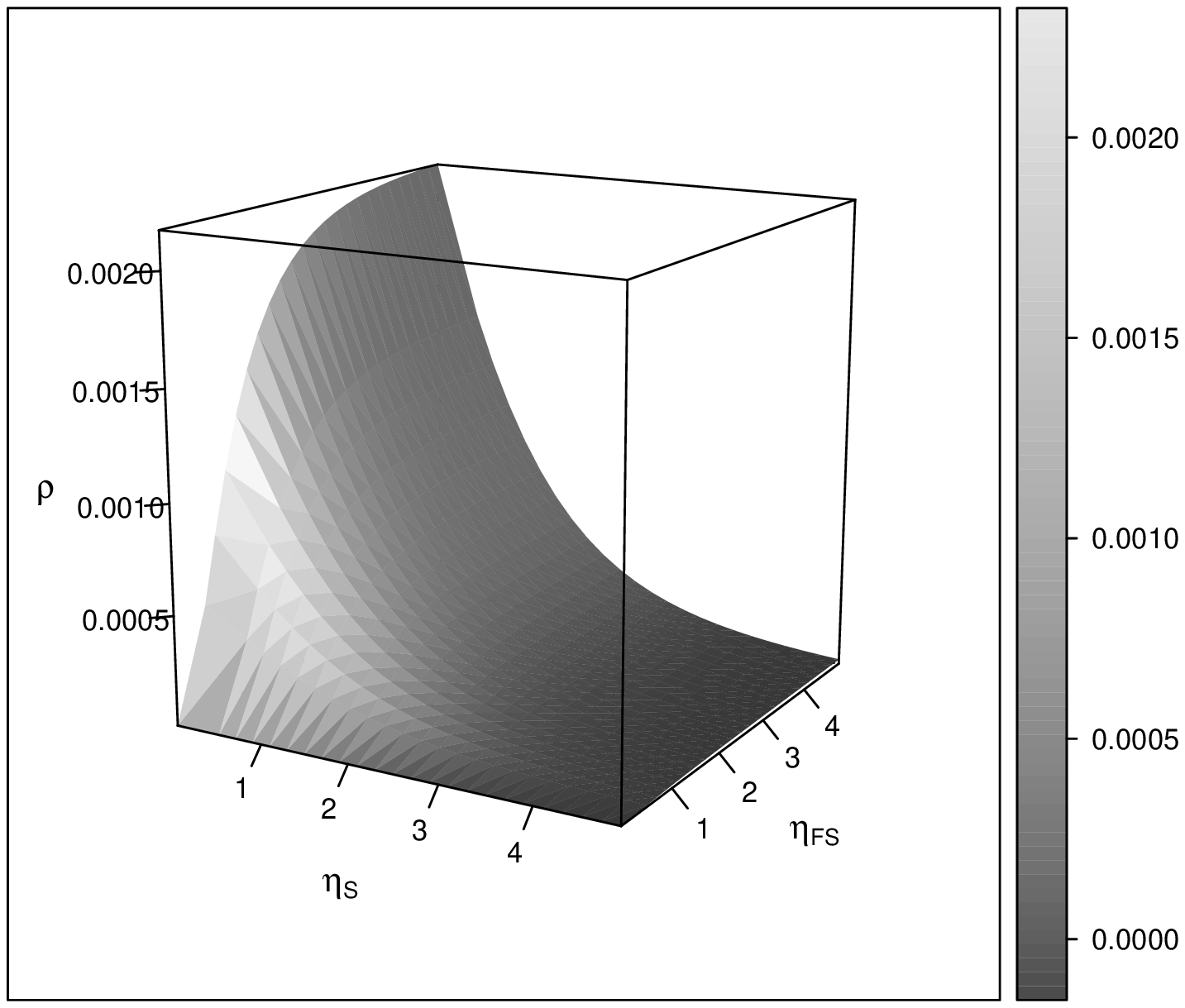}
    }
    \hspace{0.5in}
    \subfigure[q=$0.05$] {
        \label{fig:correlation_5_pp}
\includegraphics{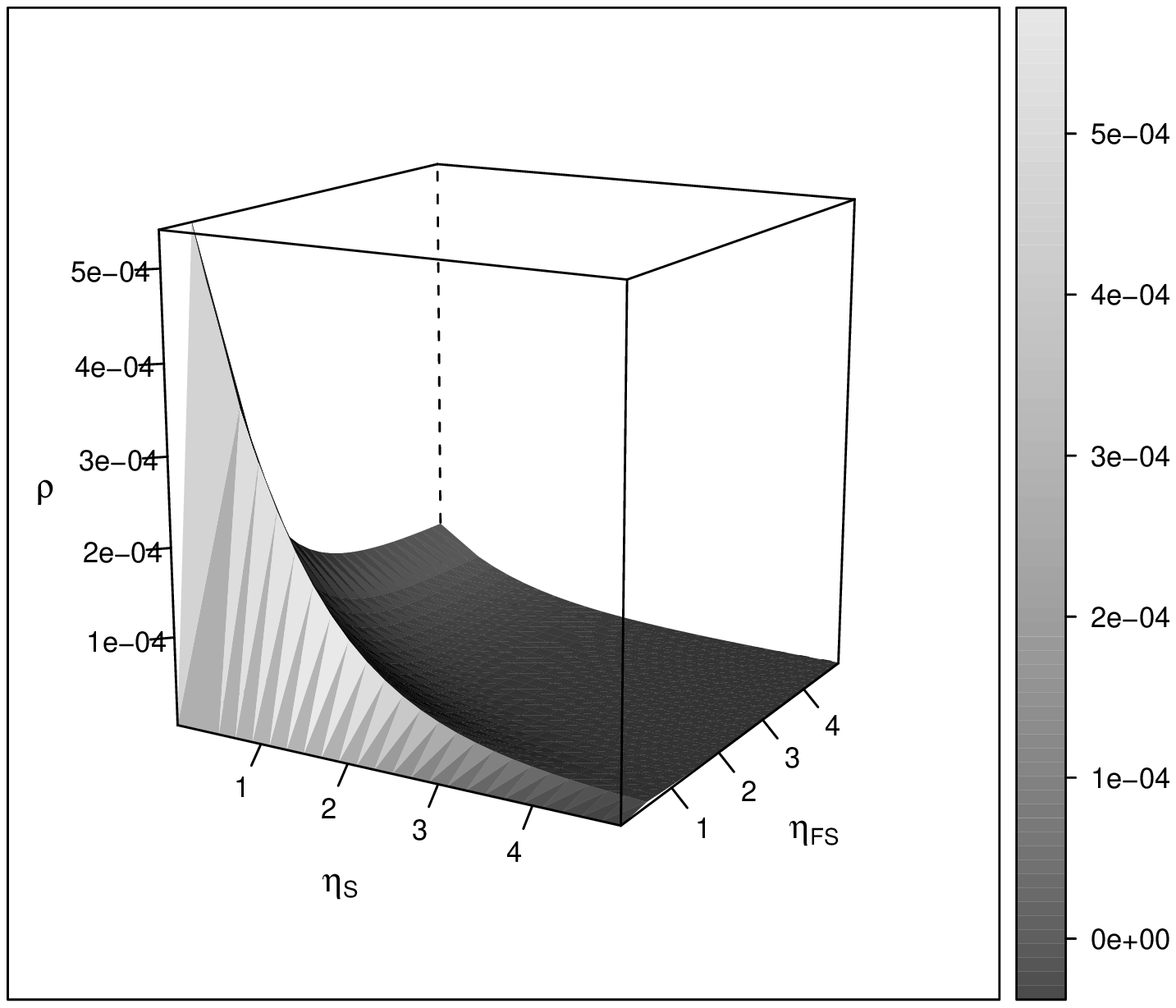}
    }
    \caption{Variation of $\rho$ for fixed marginal default
    probabilities, $\eta_S > 0, \eta_{FS} > 0$}
    \label{fig:correlation_fixed_marginal_pp}
\end{figure}
\setkeys{Gin}{width=0.8\textwidth}

 \setkeys{Gin}{width=0.4\textwidth}
\begin{figure}[!h]
    \centering
    \subfigure[q=$0.01$] {
    \label{fig:correlation_1_pn}
\includegraphics{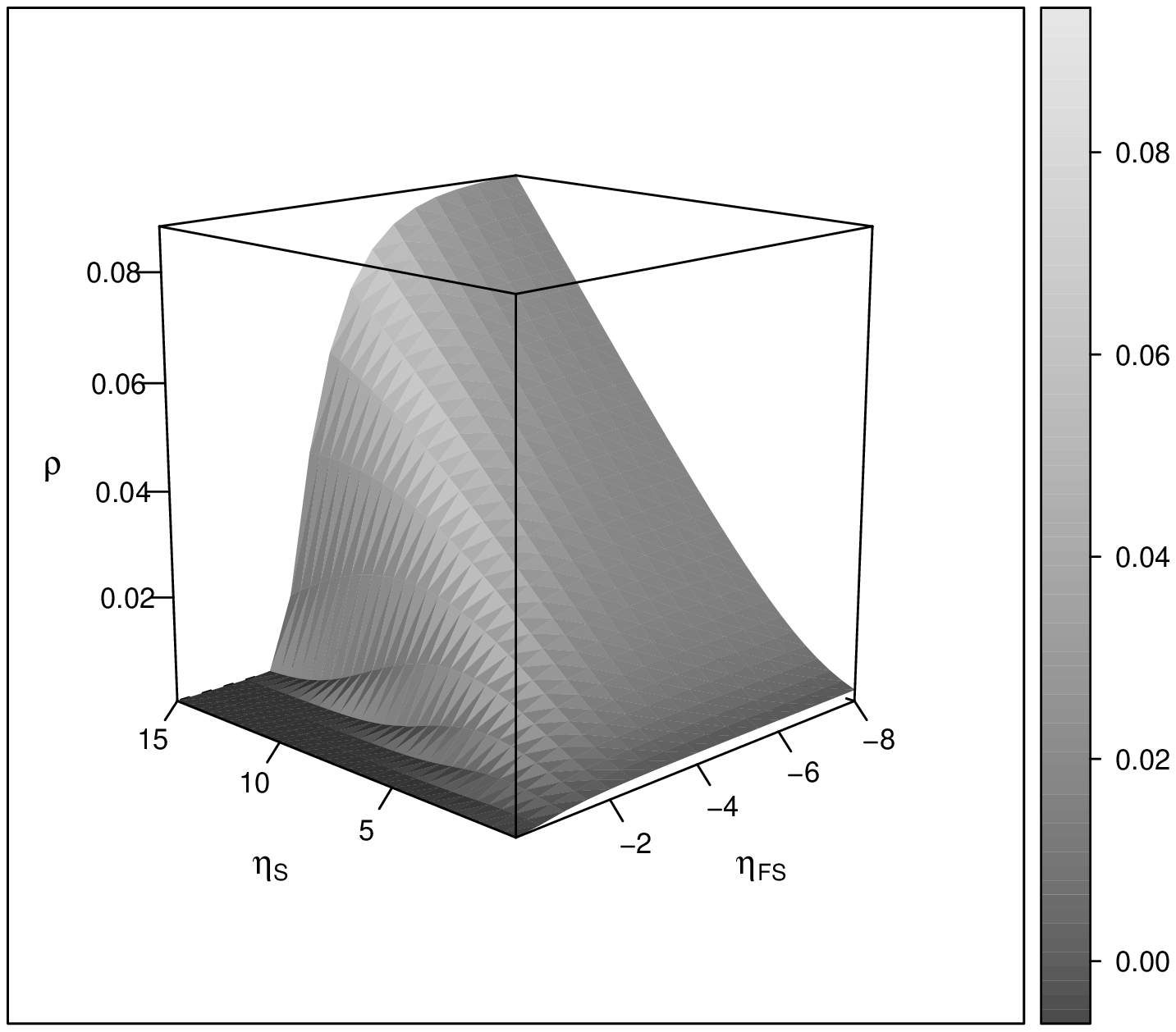}
    }
    \hspace{0.5in}
    \subfigure[q=$0.05$] {
        \label{fig:correlation_5_pn}
\includegraphics{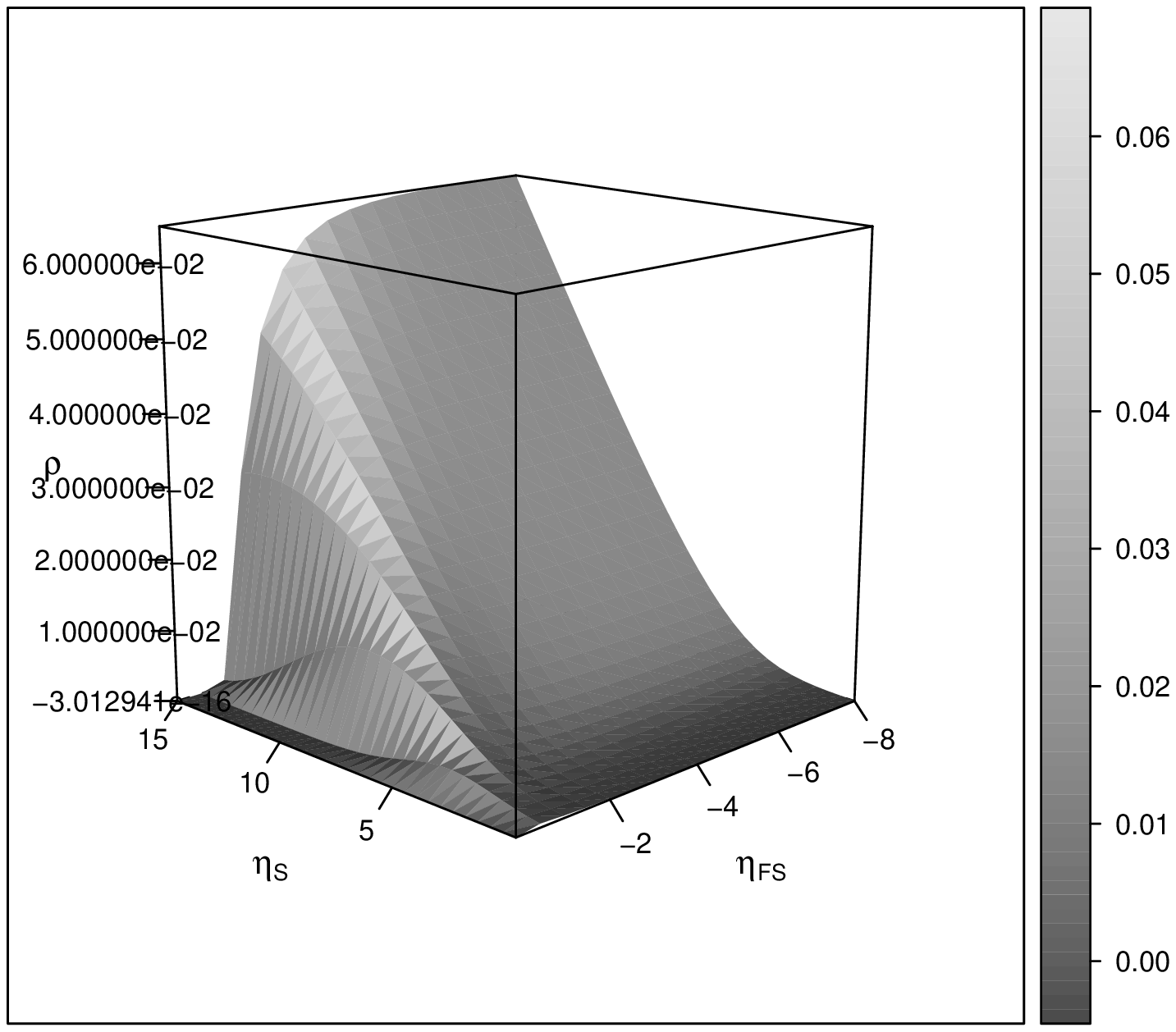}
    }
    \caption{Variation of $\rho$ for fixed marginal default
    probabilities, $\eta_S > 0, \eta_{FS} < 0$}
    \label{fig:correlation_fixed_marginal_pn}
\end{figure}
\setkeys{Gin}{width=0.8\textwidth}

\paragraph{Loss distribution for single-period model}

Figure~\ref{fig:singlesecloss} exhibits the shape and fat tails of
the loss distribution for different parameters. Here $N=$125,
$\eta_{FS}=$~-2.1, $\pr{X_i}=$0.05, with $\eta_S$ calibrated to match
the given correlation level $\rho$ and $\eta_F$ to match the given
marginal default probability. The figure shows the loss
distribution for correlation levels 0.01, 0.02, 0.05, 0.07. As
expected, the mass shifts towards the tail as correlation
increases. All the distributions are bimodal, which facilitates
having significantly fat tails.

\begin{figure}[h]
    \center
\includegraphics{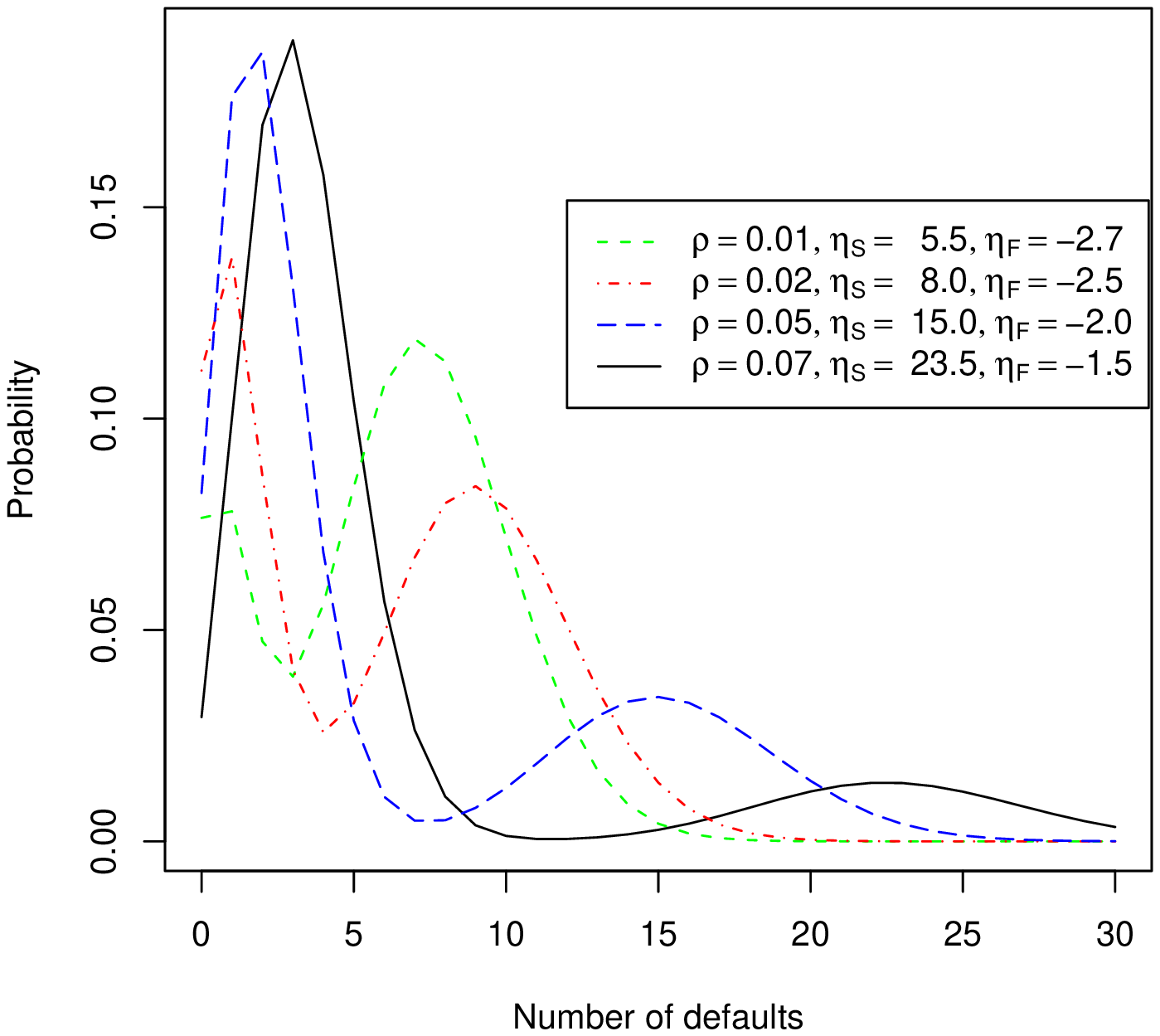}
    \caption{Distribution of number of defaults under different correlations, $\eta_{FS}=$-2.1, $\eta_S,\eta_F$ varying}
    \label{fig:singlesecloss}
\end{figure}

\paragraph{Heavy Tails for Multi-period Model}

\setkeys{Gin}{width=0.45\textwidth}
\begin{figure}[!h]
    \centering
    \subfigure[Single step loss distribution] {
\includegraphics{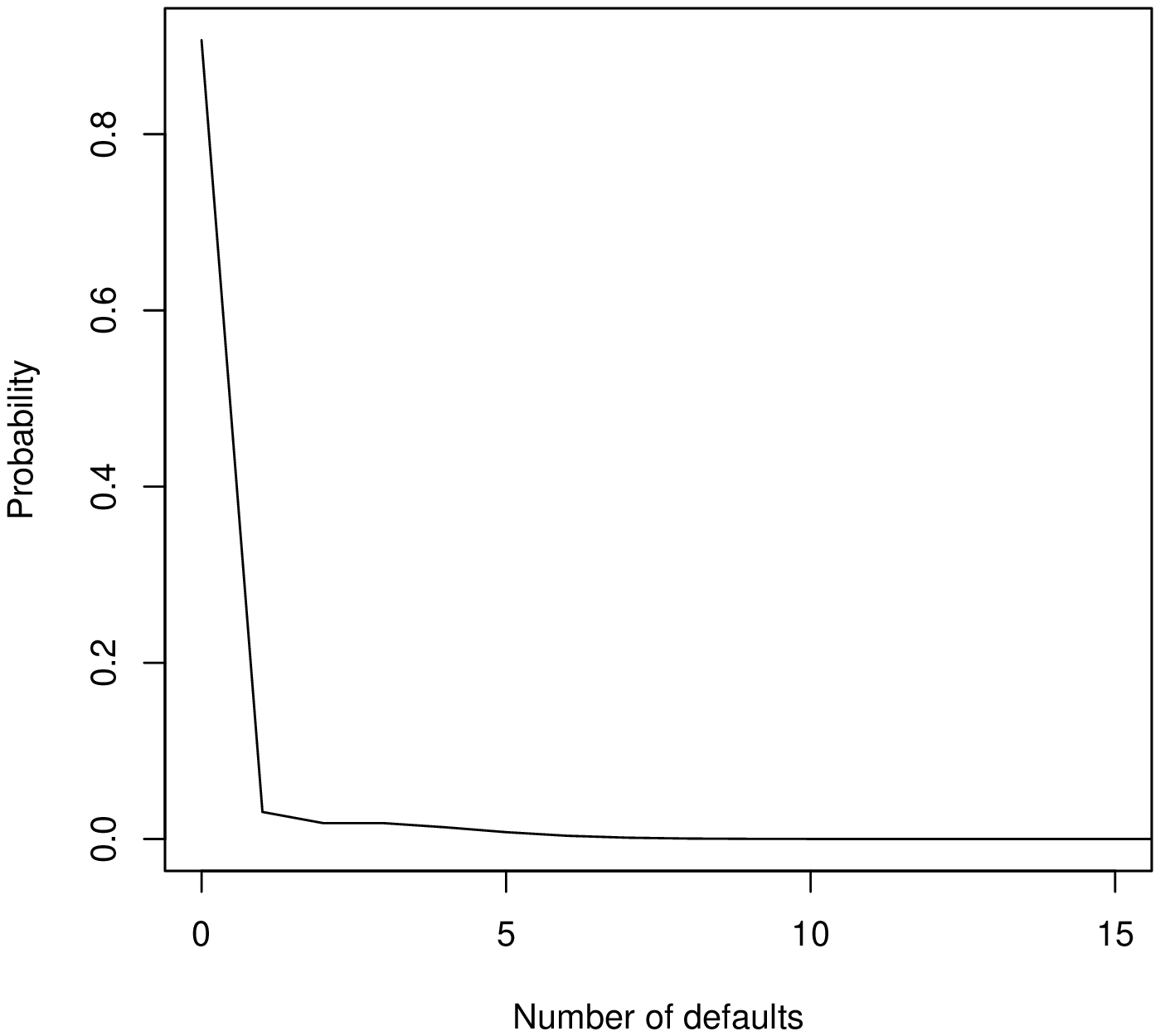}
    }
    \hspace{1in}
    \subfigure[5 steps] {
        \label{fig:wr_initial_plotter_step5}
\includegraphics{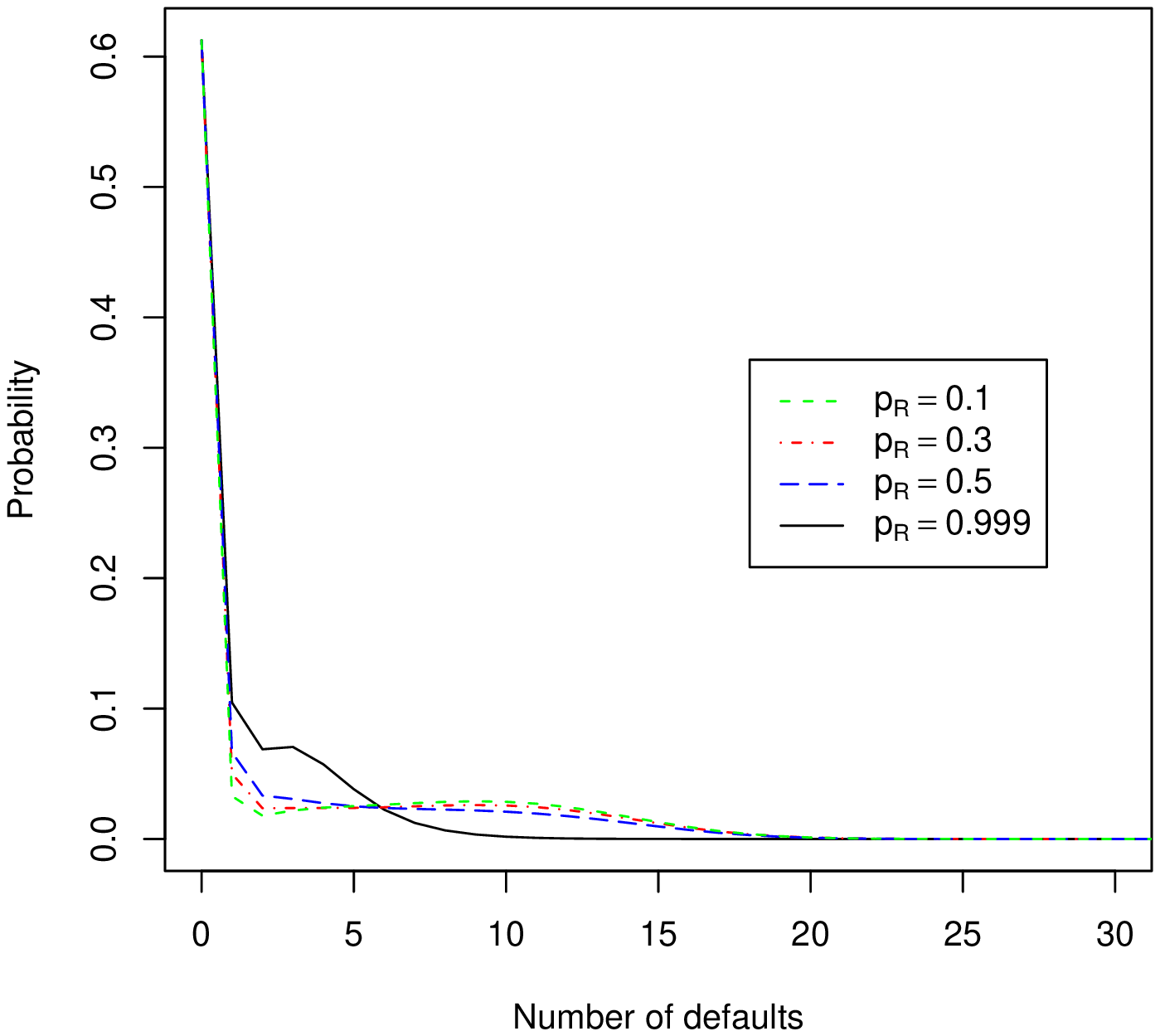}
    }
    \hfill
    \subfigure[10 steps] {
        \label{fig:wr_initial_plotter_step20}
\includegraphics{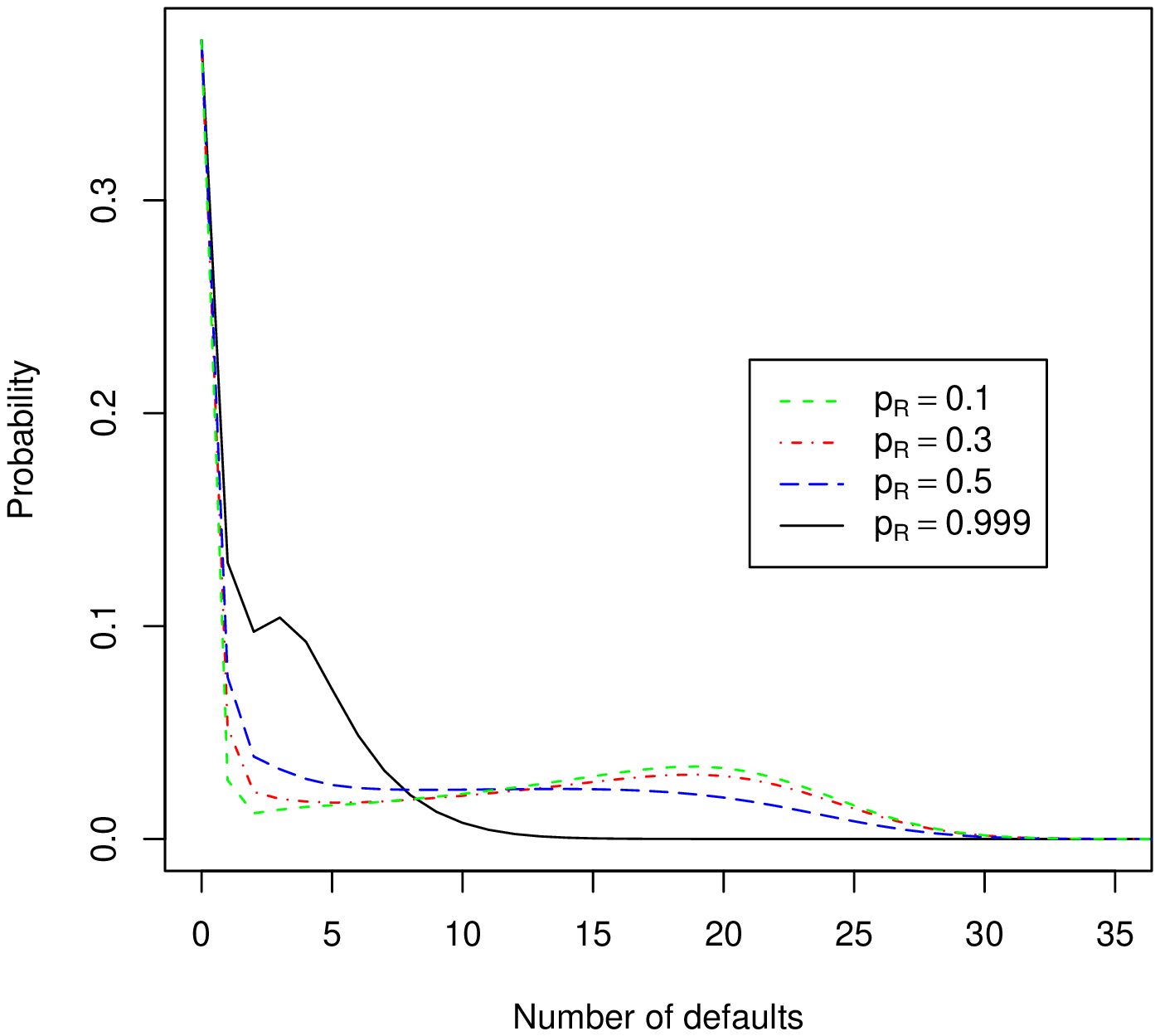}
    }
    \caption{Evolution of number of defaults for varying disappearance probabilities,
    $\eta_F$=-2.8,
    $\eta_S$=5.514,
    $\eta_{FS}$=-5,
    $\rho$=0.05,
    single step default probability 0.005,
    50 firms
    }
    \label{fig:initial_multiperiod_ld}
\end{figure}
\setkeys{Gin}{width=0.8\textwidth}

Figure~\ref{fig:initial_multiperiod_ld} shows the effect of
different parameterizations on loss distributions in a
multi-period model. Take $N=$50 firms and
$(\eta_S,\eta_{FS},\eta_F)$=
$($5.514,
-5,
-2.76$)$. The parameters are chosen
to correspond to the single-firm default probability of 0.005
and default correlation of 0.05.
The loss distribution is then calculated after 5 and 10 steps, with the
removal probabilities $p_R$ = 0.1, 0.3, 0.5, 0.999. Note that
by varying $p_R$, tail of the loss distribution can be
controlled as shown in Figure~\ref{fig:initial_multiperiod_ld}:
Increasing $p_R$ results in thinner tails, whereas lower $p_R$
values can obtain quite fat tails. As the removal
probability increases, the mass is shifted towards fewer defaults.
This is intuitive, since higher $p_R$ values lead to defaulted firms
staying in the system shorter and thus having less detrimental effect
on financially healthy firms.
Moreover, the tails become very
significant as the number of steps increase, despite a relatively
low single-step default probability of 0.005.
Figure~\ref{fig:wr_initial_plotter_step20} shows the importance of
choosing a suitable single step default
    probability. High values for this quantity results in rather strong shifts of mass towards high number of
    defaults. Therefore, for a fixed maturity, whenever the number of steps is increased, the single step default probability has
    to be scaled down accordingly.

\subsection{Comparison to Normal Copula}
\label{sec:copula}

In this section, we compare our model with the widely used
one-factor normal copula model of \cite{Li:00}. Two important
attributes are discussed: the heaviness of the tails in loss distribution, and
the correlation smile in pricing standard tranches.

\paragraph{Heavy tails for One-period Model}

One well-known deficiency of normal copulas with a constant correlation
parameter for all pairs of firms is the thinner ``tails'' in the
loss distribution than observed from market data. In comparison,
Proposition~\ref{lem:binomialdec} suggests that the
loss distribution based on our model can achieve fatter tails.

To demonstrate this, first recall that in a copula model, the default indicator $X_i$ for
firm $i$ is given by $X_i = \ind{M_i \le K}$ where
\begin{eqnarray}
    M_i &=& \sqrt{\rho_A}Y + \sqrt{1-\rho_A}\epsilon_i, \, i \in \cbr{1,\cdots,N} \label{eq:normalCopula}\\
    Y, \epsilon_i &\sim& \textrm{Normal}(0,1)\,\, i.i.d. \nonumber
\end{eqnarray}
$\rho_A = corr(M_i,M_j)$ is the {\em
asset correlation}, assumed to be the same for all pairs of firms.
Note that $K\in \reals$ implicitly specifies the
marginal default probability $\pr{X_i = 1}$ which is assumed to be
the same across all firms.
Note also that the asset correlation $\rho_A$ is different from
the {\em default correlation} $\rho = corr(X_i,X_j)$ which is
given by:
\begin{displaymath}
    \rho = \frac{{\Phi^{-1}(K)}^2 \int_{-\infty}^{\infty} \phi(y) {\Phi\left(
    \frac{K-\sqrt{\rho_A}y}{\sqrt{1-\rho_A}} \right)}^2 dy}{{\Phi(K)}^2
    {(1-\Phi(K))}^2}.
\end{displaymath}
This is an important distinction, as asset
correlation values for the normal copula
result in significantly different values of default indicator correlation.

For comparing loss distributions implied by
the two approaches, take $N=$125, $\pr{X_i = 1}=$
0.05, and two levels of $\rho=$
0.01, 0.05. For
the normal copula, these $\rho$ values lead to
to $\rho_{A}$=0.042, 0.18 respectively.
For our model, take
$(\eta_{FS},\eta_{S},\eta_{F}) = \\
($-0.95,
9.2,
-2.2$)$ and $($
-2.1,
15,
-2$)$ respectively.
These parameters are chosen so as to match the specified $\pr{X_i=1}$ and $\rho$.
For both levels of correlation, as demonstrated in Figure~\ref{fig:compareonefactor} our
model exhibits fatter tails and has smaller loss
probabilities for intermediary values. Furthermore,
the values for the loss distribution are of the same scale.
All these properties
help in correcting the deficiencies of the normal copula when
pricing CDO's, as demonstrated next.

\setkeys{Gin}{width=0.4\textwidth}
\begin{figure}[h]
    \centering
    \subfigure[$\rho=0.01$] {
\includegraphics{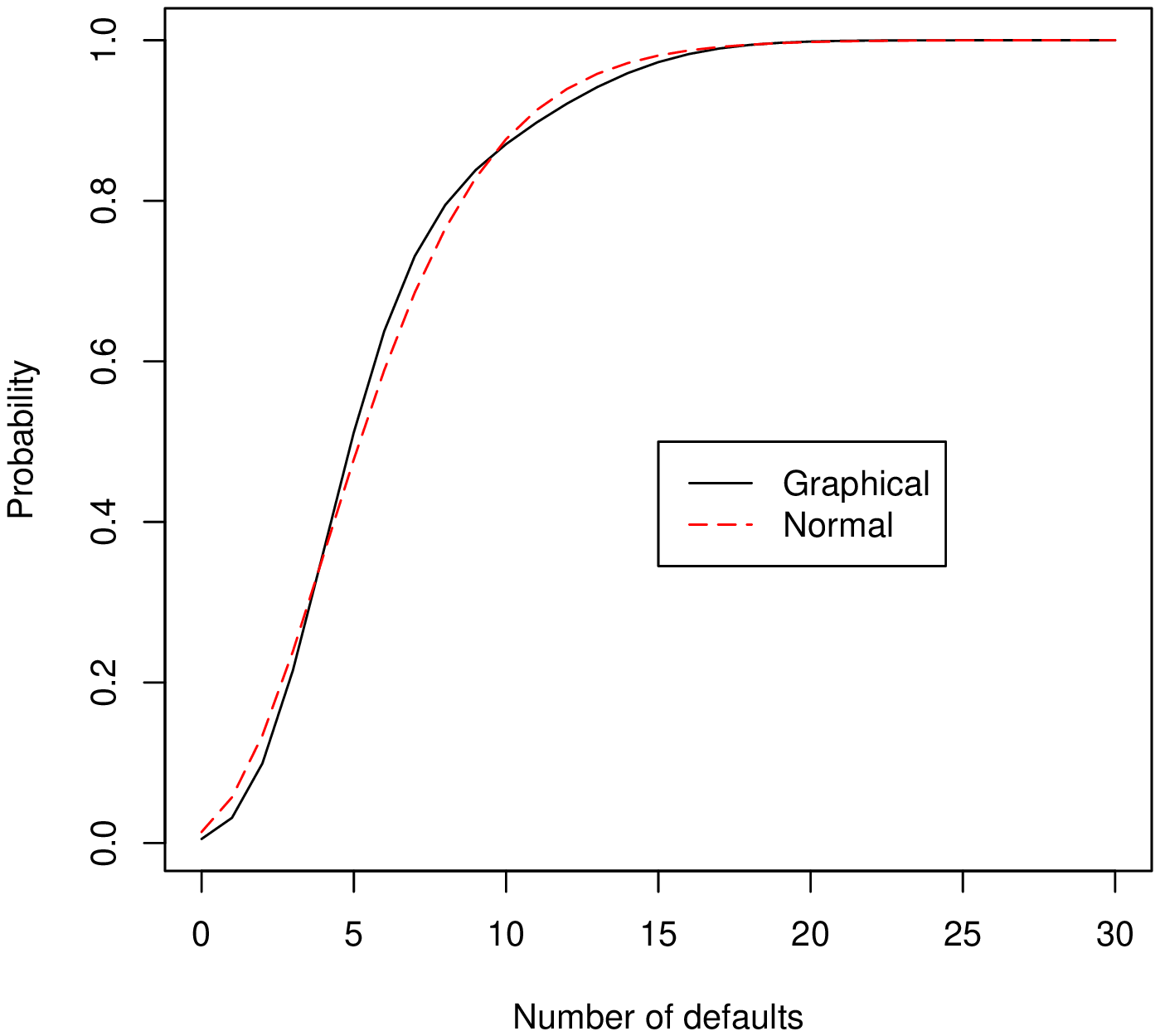}
    }
    \hfill
    \subfigure[$\rho=0.05$] {
\includegraphics{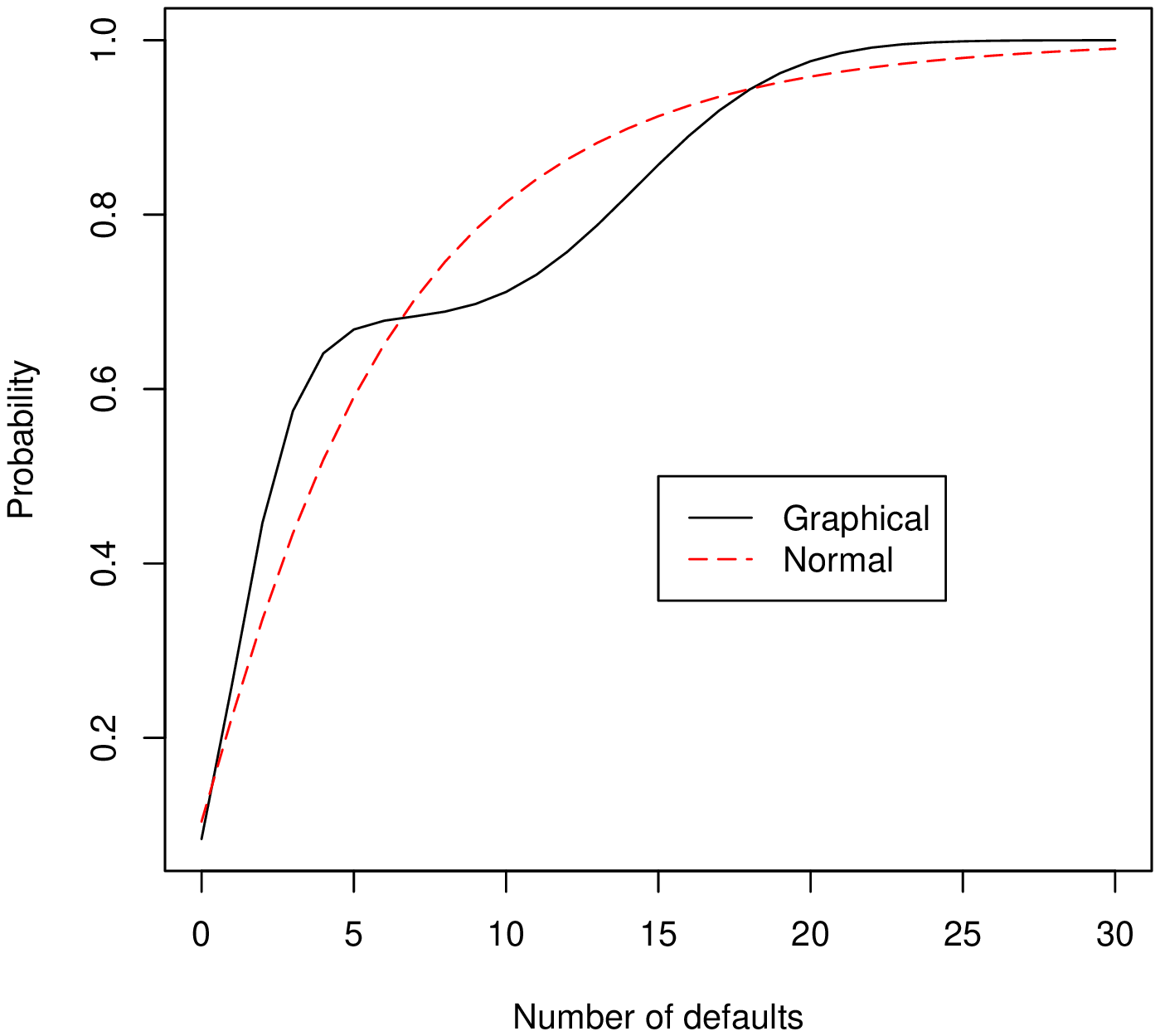}
    }
    \caption{Comparison of one-factor normal copula and single sector one-period graphical models}
    \label{fig:compareonefactor}
\end{figure}
\setkeys{Gin}{width=0.8\textwidth}


\paragraph{Correlation Smile}
\label{ssec:multiperiodgaussiancomparison}

For the normal copula, the pricing scheme
of~\cite{Li:00} and Hull and White~\cite{HullWhite:04} is utilized, where the default time
$\tau_i$ for a firm is defined through a transformation of $M_i$
in Equation~(\ref{eq:normalCopula}). More specifically, the
risk-free interest rate $r$ and the recovery rate $R$ are taken to
be constants, $\tau_i$ is assumed to be distributed exponentially
with rate $\lambda$, and $M_i$ is mapped to $\tau_i$ using a
percentile-to-percentile transformation so that for any given
realization $\hat{M}_i$:
\begin{eqnarray}
    \hat{\tau}_i &=& \frac{- \ln(1-\Phi(\hat{M}_i))}{\lambda}. \label{eq:normalTauTransformation}
\end{eqnarray}
The spreads $s_l$ are then calculated by simulating $M_i$ values
and replacing the expected values in the pricing formula in
Equation~(\ref{eq:trancheSpread}) by their respective estimators.

Recall that
given a standard tranche on CDX.NA.IG, with given observed spread
$s_l$, and known $r,R,\lambda$, it is possible to ``imply'' the
asset correlation parameter $\rho_A$ in
Equation~(\ref{eq:normalCopula}). However, it is known (e.g.
\cite{AmatoGyntelberg:05}, \cite{HagerSchoebel:06}) that implying
$\rho_A$ in such a manner across all tranches results in a ``smile'':
The mezzanine tranche has lower implied correlation compared to
the neighboring tranches. One plausible interpretation
for this kind of smile is that the normal copula model
underprices the senior tranches and overprices the equity tranche
in comparison to the mezzanine tranche.

We now demonstrate that our model has the potential to correct this smile.
To achieve this, first we calculate prices from normal copula. We then
find parameters $(\eta_F, \eta_{FS}, \eta_S, p_R)$ such that our model
matches the mezzanine tranche spread exactly with
those from normal copula while giving significantly lower
spreads for the equity tranche and higher spreads for the senior tranches.

More specifically,
take two different credit rating classes, representing high and low credit
ratings respectively, so that
\begin{itemize}
    \item the one-year default probabilities are set at
	    0.001
	    for high-rating class and
	    0.015
	    for low-rating class,
    \item the asset correlation values for the normal copula are
	    0.2
	    and
	    0.3
	    (these values correspond to default indicator correlations of
	    0.0059
	    and
	    0.0562
	    ),
    \item the recovery rate is
	    0.4,
	    the interest rate
	    0.05,
    and
	    $N$=50,
    \item the maturity of the CDO is
	    5
	    years with payment frequency
	    0.5
	    corresponding to a ten period model.
\end{itemize}
Meanwhile, for each rating class, an optimization on $(\eta_{FS}, \eta_S, p_R)$
maximizing the difference between
equity tranche spread for the normal copula and our model, and senior tranche
spreads for our model and the normal copula is run. Parameter $\eta_F$ is constrained
so that both the one-year default probability and the mezzanine spread are matched.

Figure~\ref{fig:spread_comparison} shows the output of one such optimization run.
It demonstrates that even with a flat correlation value for the
graphical model, one can obtain lower prices for
mezzanine tranche and higher for the senior tranches in comparison to the normal
copula, thus ``correcting'' the correlation smile.

\setkeys{Gin}{width=0.4\textwidth}
\begin{figure}[!h]
    \centering
    \subfigure[High-rating CDO]{
    \label{fig:class_B}
\includegraphics{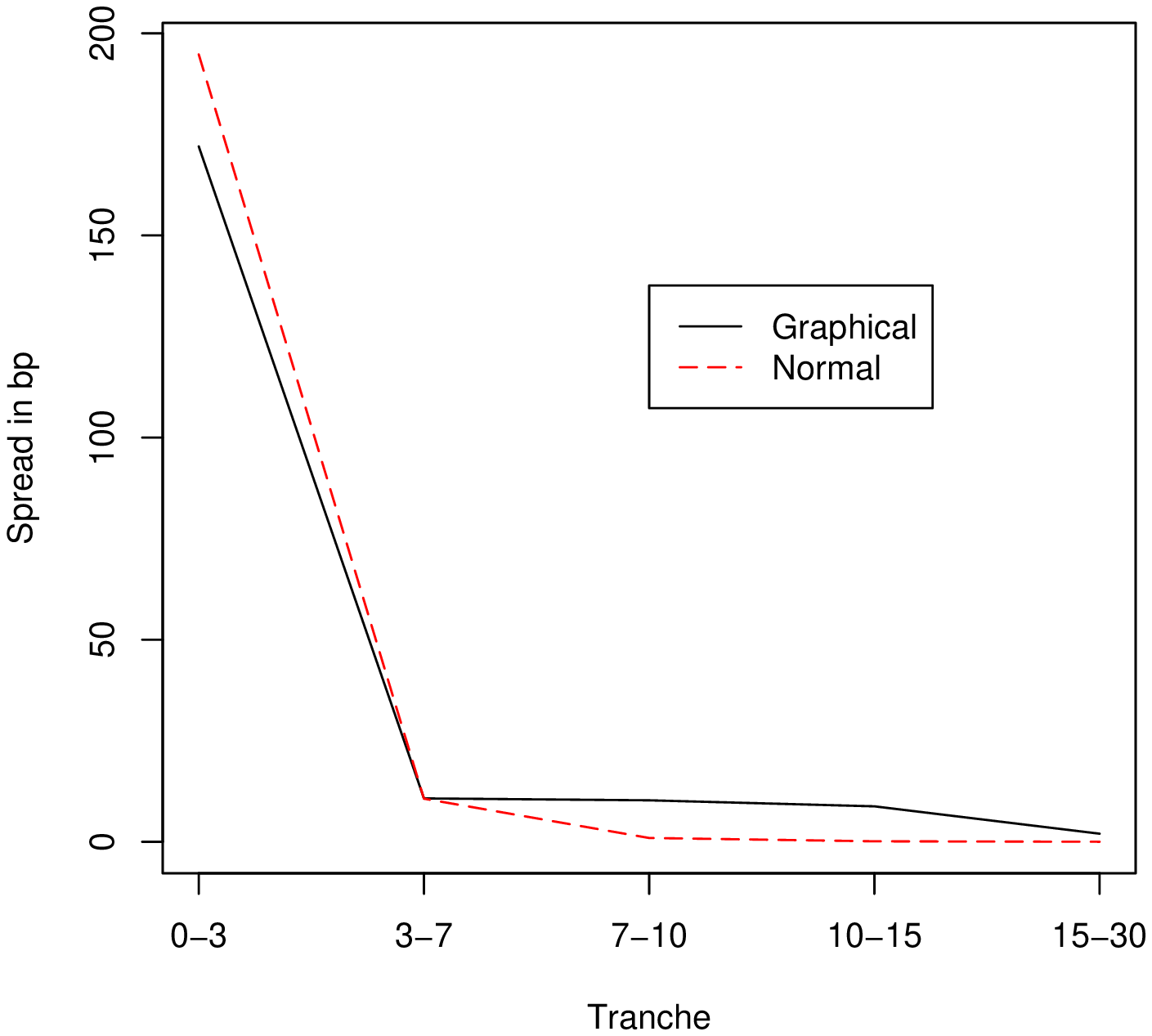}
} \hspace{0.5in}
    \subfigure[Low-rating CDO]{ 
    \label{fig:class_C}
\includegraphics{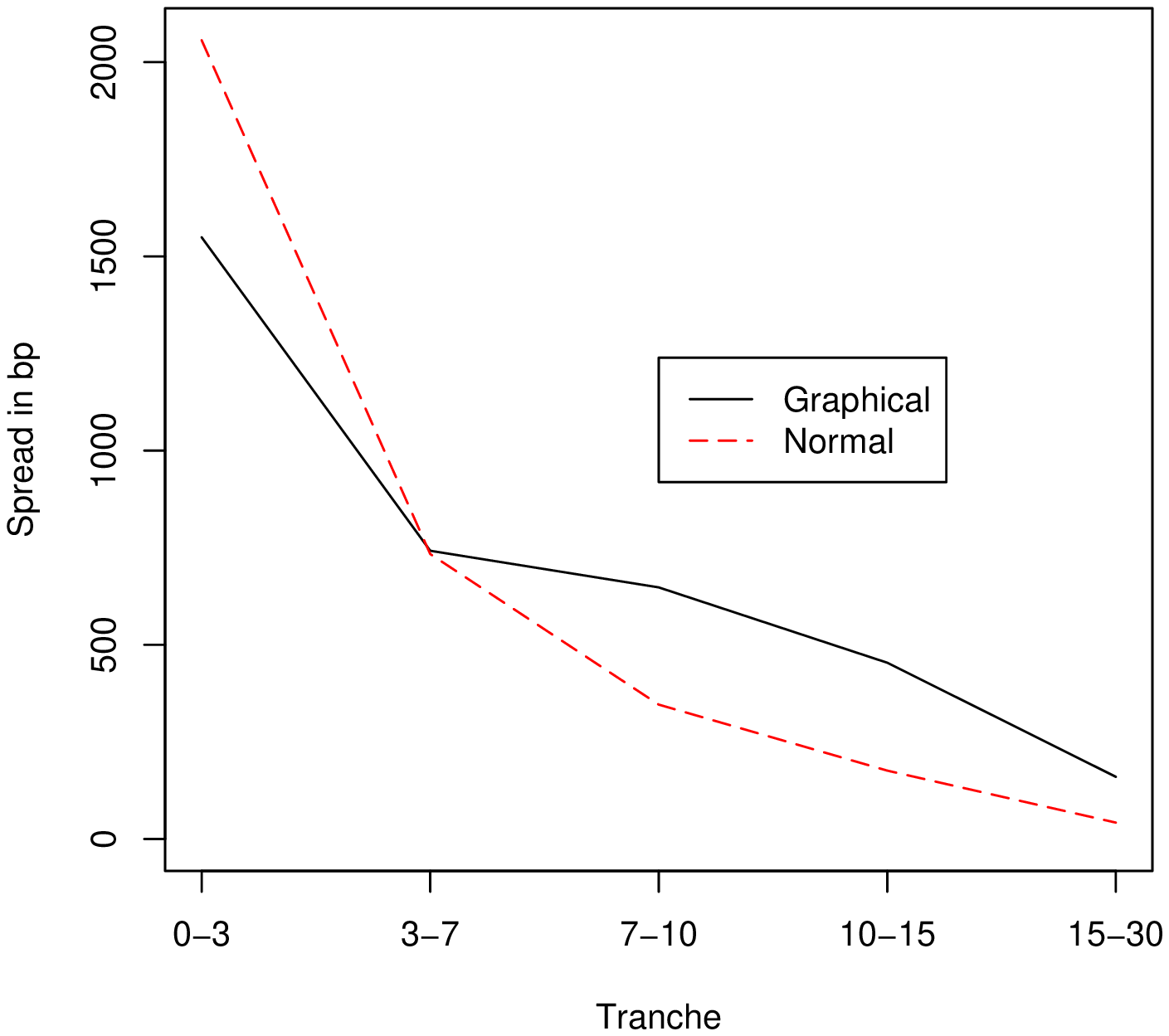}
}
    \caption{Tranche spreads for graphical and normal copula models}
    \label{fig:spread_comparison}
\end{figure}
\setkeys{Gin}{width=0.8\textwidth}

\section{Conclusion}

This paper proposes and analyzes a simple graphical model
for modelling correlated default. The graphical representation
provides an effective shorthand to depict the dependence
relationships between the $N$ firms, with the desirable conditional
independence property. The
probability distribution proposed is a {\em toric model}
which is beneficial in both
parameter estimation and simulation.
With some
homogeneity assumptions, loss distributions and CDO prices are
obtained analytically. The model
generates heavy tails in the loss distribution, and its dynamic
formulation seems promising for correcting
the correlation smile observed in one-factor normal copula.


\appendix

\section{Birch's Theorem} 
\label{app:birch}
Our proof follows that in \cite{Fulton1993}, and begins with a lemma.

\begin{lemma} \label{lem:FultonAB}
Let $A=(a_{ij})$ be a real $d \times m$ matrix of rank $d$
and $\pos(A)$ the $\R_{\geq 0}$-span of its columns $a_{\cdot 1}, \dots, a_{\cdot m}$.  Let $t^1, \dots, t^m \in \R_{>0}$ be real positive numbers and define 
\[
F: \R^d \ra \R^d\]\[
\eta \mapsto \sum_{j=1}^m t^j e^{\langle a_{\cdot j}, \eta \rangle} a_{\cdot j}
\] 
That is, $F( \eta) = A (t^1 e^{\langle a_{\cdot 1}, \eta \rangle},  \dots, t^m e^{\langle a_{\cdot m}, \eta \rangle})^{\transpose}$.  Then $F$ determines a real analytic isomorphism of $\R^d$ onto the interior of $\pos(A)$.
\end{lemma}
\begin{proof}
First, the fact that $F$ is an injective local isomorphism with image points arbitrarily close to the extreme rays of $\pos(A)$ is established.  Then the result will follow from an inductive proof that $\im(F)$ is convex.

We have $F(\eta)_i = \sum_j t^j e^{a^{\ell}_j \eta_{\ell}} a^i_j$, so 
\[
\frac{\partial F_i}{\partial \eta_k} = \sum_j t^j  e^{a^{\ell}_j \eta_{\ell}} a^i_j a^k_j
\]
so that the Jacobian is symmetric.  Moreover, the quadratic form is given by
\[
x_i( t^j  e^{a^{\ell}_j \eta_{\ell}} a^i_j a^k_j)x_k = (x_i a^i_j) t^j  e^{a^{\ell}_j \eta_{\ell}}  (a^k_j x_k) = t^j  e^{a^{\ell}_j \eta_{\ell}} (x_i a^i_j)^2
\]
which is strictly positive for $x \not= 0$
 since the $a^{\cdot}_j$ span $\R^d$, so some $a^i_j \neq 0$. 
 This shows that $F$ is a local isomorphism.  To show it is injective, it is sufficient to check that $F$ is injective on a line, and by a change of coordinates this reduces the problem to the case $d=1$.  In this case, the $a^{\cdot}_j$ are scalars $a_j$ and $F$ sends $\eta \in \R$ to $t^j a_j e^{a^i_j \eta_i}$, with strictly positive derivative as above.  $\pos(A)$ is either $[0,\infty)$, $(\infty, 0]$, or $(-\infty, \infty)$ depending on the signs of the $a_j$, and $F$ is an isomorphism of manifolds.  Thus the $d=1$ case shows $F$ is injective and will also serve as the base case for our induction.

By grouping the $a^{\cdot}_j$ and changing the $t^j$, one may assume no two $a^{\cdot}_j$ lie on the same ray.  Suppose $a^{\cdot}_1$ generates an extreme ray of $\pos(A)$.  Then one may choose $v$ such that $\langle v, a^{\cdot}_1 \rangle=0$ but $\langle v , a^{\cdot}_1 \rangle <0$ for $j \neq 1$.  Then $F(\lambda v + \eta) = t^1 e^{\langle \eta , a^{\cdot}_1 \rangle} a^{\cdot}_1 +  \cdots + t^1 e^{\langle \eta , a^{\cdot}_1 \rangle + \lambda \langle v , a^{\cdot}_m \rangle}   a^{\cdot}_m$, so \[
\lim_{\lambda \ra \infty} F( \lambda v + \eta) = t^1 e^{\langle \eta , a^{\cdot}_1 \rangle} a^{\cdot}_1
\]
so that one can approach any point on $\R_{\geq 0} a^{\cdot}_1$ arbitrarily closely by adjusting $\eta$.

It remains to show that the image of $F$ is convex.  Suppose $\im(F)$ is convex for $d-1$, and let $L$ be a line in $\R^d$; to show $\im(F)$ is convex for $d$, one must show that $L \cap \im(F)$ is connected or empty.  One can write $L = \pi^{-1} q$ for a suitable projection $\pi:\R^d \ra \R^{d-1}$ and point $q \in \R^{d-1}$, and 
\[
\im(F) \intersect L = F(\R^d) \intersect \pi^{-1}(q) = F((\pi \circ F)^{-1} (q)).
\]
By a linear change of coordinates in $\R^d$, one may assume that $\pi$ is projection onto the first $d-1$ coordinates.  Let $\rho$ denote the projection to the last coordinate, $\eta \mapsto \eta_d$.  Let $G = \pi \circ F$ and for $y \in \R$ let $G_y$ be the restriction of $G$ to $\rho^{-1} y$.  

This defines a map $G_y: \R^{d-1} \ra \R^{d-1}$.  $G_y(\eta_1, \dots, \eta_{d-1}) = \sum_j s^j e^{\eta_{\hat{d}} \cdot a^{\hat{d}}_j} a^{\hat{d}}_j$ with $s^j = t^j e^{y a^d_j} >0$.  Still the columns of its defining matrix $\hat{A}$ ($A$ without its last row) span $\R^{d-1}$, so $G_y$  meets the hypotheses of the theorem for $d-1$.  Thus each $G_y$ is an injective map onto $\interior(\pi (\pos(A)))$.  Then for each $q$ in $\interior(\pi(\pos(A)))$, the projection $G^{-1}(q) = (\pi \circ F)^{-1}(q)$ to $\R$ induced by $\rho$ is a bijection.  So the intersection is connected.
\end{proof}

Now Lemma~\ref{lem:FultonAB} can be applied to polytopes.

\begin{prop}
Let $A=(a_{ij})$ be a real $d \times m$ matrix, and $K$ be the convex hull of its columns $a_{\cdot 1}, \dots, a_{\cdot m}$.  
Further require that the $a_{\cdot j}$ not be contained in any hyperplane.
Let $t_1, \dots, t_m \in \R_{>0}$ be real positive numbers and define 
\[
H: \R^d \ra \R^d\]\[
\eta \mapsto \frac{1}{Z(\eta)}\sum_{j=1}^m t^j e^{\langle a_{\cdot j}, \eta \rangle} a_{\cdot j}
\] 
where $Z(\eta) = t_1 e^{\langle a_1^{\cdot},\eta \rangle} + \cdots t_m e^{\langle a_m^{\cdot},\eta \rangle}$.  Then $H$ defines a real analytic isomorphism of $\R^d$ onto the interior  of $K$.
\end{prop}
\begin{proof}
Form the cone over $K$ in $\R^{d+1}$, letting $\tilde{a_j} = a^1_j, \dots, a^d_j, 1$.  Let \[
\tilde{F}: \R^{d+1} \ra \R^{d+1}\]\[
\eta,\eta_{d+1} \mapsto \sum_j t^j e^{\langle a_j, \eta \rangle}e^{\eta_{d+1}} \tilde{a_j}
\]
since the $a_j$ were not contained in any hyperplane, after lifting they still span $\R^{d+1}$.  Then by Lemma \ref{lem:FultonAB}, $\tilde{F}$ maps $\R^{d+1}$ isomorphically onto $\interior(\pi(\pos(\tilde{A})))$.  The last coordinate of $\tilde{F}$ is $\sum_j t^j e^{\langle a_j, \eta \rangle}e^{\eta_{d+1}}$; this is equal to $1$ when $\eta_{d+1} = - \log (Z(\eta))$.  Thus $H(\eta) = \tilde{F}(\eta, - \log (Z(\eta)))$ maps $\R^d$ isomorphically onto $\interior(K)$.
\end{proof}
Note that if all $a_j$'s lie in a hyperplane (such as when the column sums of $A$ are equal), coordinates can be changed so that the last row of $A$ is all ones and the $a^{\hat{d}}_j$ do not lie in a hyperplane.

\section{Proof of Corollary~\ref{cor:max_likelihood}}
\label{app:max_likelihood}

The proof is from \cite[Proposition 1.9]{PachterSturmfels:05}.
First assume that $P_{\bullet}$ is a vector of marginal default probabilities and linear correlations with $M+|E|+1$
coordinates, where the last coordinate is $1$.
This implies that there exists a $q \in \cbr{p: \sum_{w=1}^{2^M} p_w = 1}$ such that $A_G q = P_\bullet$. Now take $u=q$.
Equation~(\ref{eq:likelihood_function}) can equivalently be written as:
\begin{displaymath}
	\textrm{Maximize } \theta^{A_G u} \textrm{ subject to } \theta \in \reals^{M+|E|}_{>0} \textrm{ and } \sum_{j=1}^{2^M} \theta^{a_j} = 1
\end{displaymath}
where:
\begin{displaymath}
	\theta^{A_G u} :=\prod_{i=1}^{M+|E|}\prod_{j=1}^{2^M} \theta_i^{a_{ij}u_j} = \prod_{i=1}^{M+|E|} \theta_{i}^{a_{i1}u_1 + \cdots + a_{i,2^M}u_{2^M}} \textrm{ and } \theta^{a_j} = \prod_{i=1}^{M+|E|} \theta_i^{a_{ij}}
\end{displaymath}

Writing $b=A_G u$ for the sufficient statistic, our optimization problem is:
\begin{equation}
	\textrm{Maximize } \theta^{b} \textrm{ subject to } \theta \in \reals^{M+|E|}_{>0} \textrm{ and } \sum_{j=1}^{2^M} \theta^{a_j} = 1
\label{eq:optimization}
\end{equation}

Using $f(\theta):= (p_1(\log \theta), \cdots, p_{2^M}(\log \theta))$ with $p_w(.)$ given by Equation~(\ref{eq:Ising}):
\begin{prop}
	Let $\hat p = f(\hat \theta)$ be any local maximum for the problem~(\ref{eq:optimization}). Then:
	\begin{displaymath}
		A_G \hat p = b
	\end{displaymath}
\end{prop}
\begin{proof}
Introduce a {\em Lagrange multiplier} $\lambda$. Every local optimum of (\ref{eq:optimization})
is a critical point of the following function in $M+|E|+1$ unknowns $\theta_1, \cdots, \theta_{M+|E|}, \lambda$:
\begin{displaymath}
	\theta^b + \lambda \left( 1 - \sum_{j=1}^{2^M} \theta^{a_j} \right)
\end{displaymath}

Apply the scaled gradient operator
\begin{displaymath}
	\theta \cdot \nabla_\theta = \left( \theta_1 \frac{\partial}{\partial \theta_1}, \cdots, \theta_{M+|E|} \frac{\partial}{\partial \theta_{M+|E|}} \right)
\end{displaymath}
to the function above. The resulting critical equations for $\hat \theta$ and $\hat p$ state that
\begin{displaymath}
	(\hat \theta)^b \cdot b = \lambda \sum_{j=1}^{2^M} (\hat \theta)^{a_j} a_j = \lambda A_G \hat p
\end{displaymath}
This says that the vector $A_G\hat p$ is a scalar multiple of the vector $b=A_Gu$. Since the
last row of $A_G$ is assumed to be $(1,\cdots,1)$, and last element of $b$ to be $1$, $A_G \hat p = b$.
\end{proof}

As the matrix $A_G$ is assumed to have full row rank,
the proof of Corollary~\ref{cor:max_likelihood} follows from Theorem~\ref{prop:representation}, which states
that the parameters satisfying $A_G p(\eta) = P_\bullet$ are unique.
\end{document}